\documentclass{article}
\usepackage{amsmath}
\usepackage{amssymb}
\usepackage{amsthm}
\usepackage{natbib}

\usepackage{geometry}
\usepackage{hyperref}
\usepackage{algorithm}
\usepackage{graphicx} 
\usepackage{algpseudocode}
\usepackage[capitalise,nameinlink
    ,noabbrev]{cleveref}
  \hypersetup{
    colorlinks=true,
    linkcolor=blue,
    citecolor=blue
}
\usepackage{autonum}
\usepackage{xcolor}
\usepackage{multirow}  
\usepackage{setspace}
\usepackage{authblk}

\newcommand{\T}{\top}

\newcommand{\op}{o_{\mathbb{P}}}

\newcommand{\dd}{\mathrm{d}}

\newcommand{\E}{\mathbb{E}}

\newcommand{\R}{\mathbb{R}}

\newcommand{\mf}{\mathbf}

\newtheorem{theorem}{Theorem}[section]
\newtheorem{assumption}{Assumption}[section]

\newtheorem{proposition}{Proposition}[section]

\newtheorem{remark}{Remark}[section]
\newtheorem{lemma}{Lemma}[section]
\newtheorem{corollary}{Corollary}[section]
\definecolor{darkgreen}{rgb}{0.0, 0.4, 0.1} 

\title{Measuring deviations from spherical symmetry}
\author{Lujia Bai, Holger Dette}
\affil{Fakult\"at f\"ur Mathematik,\\ Ruhr-Universit\"at Bochum, 44780 Bochum, Germany}

\begin{document}

\maketitle
\begin{abstract}
Most of the work on checking  spherical symmetry  assumptions on the distribution of a $p$-dimensional random  vector $Y$
focuses on  testing for the null hypothesis of exact spherical symmetry. In this paper, we take a different point of view and
consider the problem of testing if the deviation from spherical symmetry exceeds a given threshold. For this purpose, we  propose a novel measure for the deviation from spherical symmetry, which is based on  the minimum distance between the   distribution
of the vector $\big (|Y|, Y/ |Y| )^\top $
and its best approximation by a distribution of a vector
$\big (|Y_s|, Y_s/ |Y_s | )^\top $
, where $Y_s$ follows a spherical distribution. We develop estimators for the minimum distance with corresponding statistical guarantees (provided by asymptotic theory) and use these results to develop statistical inference. Our approach assumes neither  Gaussianity nor the existence of moments, and we  demonstrate the applicability of our approach by means of a simulation study and a real data example.
\end{abstract}

\section{Introduction}
   \def\theequation{1.\arabic{equation}}	
   \setcounter{equation}{0}
 
Spherical contoured models include multivariate Gaussian distribution, multivariate $t$- distribution and multivariate stable distribution  find their applications in  vector cardiology \citep{mardia1999directional},
 portfolio  theory 
\citep{gupta2013elliptically}, Mahalanobis distances \citep{luis2005}. There has been a wide interest in testing for the spherical symmetry, or a wider class of elliptical symmetry, see for example  \cite{Smith01011977}, \cite{ludwig1991}, \cite{FANG199334}, \cite{KOLTCHINSKII1998228}, \cite{SCHOTT2002}, \cite{HP2007}, \cite{liangetal2008}, \cite{PG2008},  \cite{Henzeetal2014}, \cite{holzman2020}, \cite{babic2021}, \cite{huang2023multivariate}, \cite{banerjee2024consistent} among others.  A common feature of most of the cited references is that the authors use different approaches to construct tests for the null hypothesis of exact spherical symmetry. However, this assumption is rarely met because one does not believe that it is  exactly (with mathematical equality) satisfied, but with the hope that the deviation from the null hypothesis is small such that (optimal) inference  under the assumption of spherical symmetry is still reliable.  Thus, strictly speaking, there are many situations where tests for exact sphericity are performed, although it is clear that this hypothesis is at most ``approximately'' satisfied.

In this paper, we take a different  point of view on the problem of validating the assumption of spherical symmetry on the distribution of a random vector $Y$. Instead of testing for exact spherical  symmetry, we propose a measure  of deviation from spherical symmetry.
This measure is based on  the minimum distance between the   distribution 
of the vector $\big (\|Y\|, Y/ \|Y\| )^\top $
and its best approximation by a distribution of a vector 
$\big (\|Y_s\|, Y_s/ \|Y_s \| )^\top $
corresponding to  a vector $Y_s$ with a 
spherical distribution, where $\| \cdot \|$ denotes the Euclidean norm on $\mathbb{R}^p$
(note  that a random variable $Y_s$ has  a spherical distribution 
if and only if $\|Y_s\| $ and  $ Y_s/ \|Y_s\| $ are independent and $ Y_s/ \|Y_s\| $ is uniformly distributed on the $p$-dimensional sphere $\mathbb{S}_{p-1}$).
In Section \ref{sec2} we derive  an explicit expression for this minimum distance, which is essentially the difference between the (squared) $L^2$-norm of the density of the vector $\big (\|Y_s\|, Y_s/ \|Y_s \| )^\top $ and the (squared) $L^2$-norm of the density  of the random variable $\|Y_s\|$  multiplied with the surface area of $\mathbb{S}^{p-1}$. Consequently,  the estimation problem of the minimum distance boils down to estimating the squared  $L^2$-norm of the density of a distribution on  the sphere and a distribution on $\mathbb{R}^+$. While 
the problem  of estimating  the integrated
squared density 
on $\mathbb{R}^p$
has been extensively studied from various perspectives including optimality and adaptivity \citep[see][among many others]{hall1987estimation,bickel1988estimating,LaurantMassart2000,GineNickl2008}, the corresponding problem  on the sphere have  not found much attention so far. 
We construct  an estimator  for the minimum distance from a  sequence of independent identically distributed observations  combining an estimator for the integrated squared  density of the joint distribution of $ Y/ \|Y\| $  and $\|Y\| $   with an estimator  for the integrated squared  density of $\|Y\| $   on the non-negative line.
We prove asymptotic normality of an appropriately centered and normalized version. Interestingly, from the theoretical perspective, there appear to be different scalings, depending on whether the distribution is spherically symmetric or not. In the latter case, the difference of the estimators converges at a rate 
$\sqrt{n}$ as in the case estimating the integral of the squared density \citep[see, for example,][]{bickel1988estimating}. On the other hand, if the distribution of $Y$ is in fact spherically symmetric, 
the estimator converges at a rate 
$\sqrt{n(n-1)h/{\kappa^{(p-1)/2}}}$,
where $h \to 0 $ and $\kappa \to \infty $ denote smoothing parameters required for the estimation of the integrated squared densities of 
$\|Y\|$ and  $Y / \|Y\|$, respectively. 

Based on  these results, we develop several statistical applications. In Section \ref{sec32} we derive an asymptotic  confidence interval for the deviation of the distribution from sphericity. Moreover, we also construct tests for the  hypothesis  that the deviation from sphericity exceeds a given threshold. For example, if ${\cal M}^2$ denotes the measure of deviation we derive a consistent and  asymptotic level $\alpha$-test for the hypotheses $H_0^{\rm eq}: {\cal M}^2 > \Delta $ versus $H_1^{\rm eq}: {\cal M}^2 \leq \Delta $, which allows to decide for {\it approximate sphericity} at a controlled type $I$ error. These procedures require the estimation of a complicated asymptotic variance, which we address by a Jackknife procedure. Moreover, as an alternative, we also develop in Section \ref{sec4} pivotal inference for the measure ${\cal M}^2$, which does not require any estimation of asymptotic variances. Our approach is based on the self-normalization principle \citep[see]
[]{shao2015}, which cannot be used directly in the present context. In particular, we have to derive the weak convergence of a process of non-parametric U-statistics with random elements on the sphere, and adapt it to obtain a pivotal statistic for the estimator of ${\cal M}^2$. 
Finally, we demonstrate the applicability of our approach through a small simulation study and a real data example.

\section{A measure of deviation and its estimation}
\label{sec2} 
   \def\theequation{2.\arabic{equation}}	
   \setcounter{equation}{0}
   
Let $Y$ denote a $p$-dimensional random variable and  define $U = \|Y\|$, $V = Y/U$, where $\|\cdot\| $ denotes the Euclidean norm. The density function of $Y$ can be written as 
\begin{align}
\label{det0}
    f_Y(y) = f_{U,V}(u,v) = f_U(u) f_{V|U}(u,v) , 
\end{align}
where  $u=\|y\|$ and $v=y/\|y\|$. Throughout this paper, we will use $f$ for $f_{U,V}$. 
The distribution of 
    $Y $ is called spherically distributed if and only if its  density function
    can be represented in the form 
\begin{align}
\label{def:sphere}
    f_Y(y) = g(\| y\|), 
\end{align}
for some function $g$.
Consequently, observing  \eqref{det0}, it follows that  
$Y$ has a spherically symmetric distribution if and only if  $U=\|Y \|$ and $V=Y/\|Y\|$ are independent and $V=Y/\|Y\|$ is uniformly distributed on the sphere with (constant)  density 
\begin{align}
    \label{det101}
f_{0}(v) =\omega_{p-1}^{-1} I(\|v\| = 1),
\end{align}
 where $\omega_{p-1} = 2\pi^{p/2}/\Gamma(p/2)$ denotes the surface area of the sphere $\mathbb{S}^{p-1}$. Note that any spherical distribution of a $p$-dimensional random variable corresponds to a density  of the form $f_U (u) f_0(v)$ for some density $f_U$ on $\mathbb{R}^+$, where $f_0$ is given by \eqref{det101}.
 Consequently, we define  
\begin{align}
\label{det103}
\mathcal{M}^2 = \min_{h} \int_{\R^+}\int_{\mathbb{S}^{p-1}} (f(u,v) - h(u) f_{0}(v))^2 \dd u\  \omega_{p-1}(\dd v),
\label{det102}
\end{align} 
as a measure for the  deviation from spherical symmetry, where 
 $\omega_{p-1}(\dd v) $ is the area element of $\mathbb{S}^{p-1}$, where the minimum is taken over the set of densities on $\mathbb{R}^+$. Note that $\omega_{p-1} =  \int_{\mathbb{S}^{p-1}} \omega_{p-1}  (\dd v)$.  Our first result provides an explicit solution of this optimization problem,  and is proved in \cref{proof2.1}. 

 \begin{proposition}
   \label{mindist} 
The minimum in \eqref{det103} is obtained for the marginal density of $(U,V) = ( \|Y\|, Y/\|Y\|)$, that is 
$ h^*(u):=  f_U(u) = \int_{\mathbb{S}^{p-1}} f(u,v) \ \omega_{p-1}(\mathrm dv)$, and  given by  
  \begin{align}
 \mathcal{M}^2 = \int_{\R^+}\int_{\mathbb{S}^{p-1}} f^2 (u,v) \mathrm  du \ \omega_{p-1}(\mathrm  dv)  - \omega_{p-1}^{-1} \int_{\R^+} f^2_U(u)\mathrm  du . \label{eq:optdist}
\end{align} 
 \end{proposition}

 \noindent
Let $Y_1, \ldots, Y_n$  denote independent identically distributed $p$-dimensional   random variables and define
$U_i =\|Y_i \|$; $V_i = Y_i / \| Y_i\| $ ($i=1,\ldots , n)$.
 By Proposition \ref{mindist} the problem of estimating  the minimum distance ${\cal M}^2$ boils down to the estimation of the integrated squared densities of the random variables $(U_i,V_i)$ and $U_i$. The latter  can be estimated by standard methods \citep[see, for example][among many others]{hall1987estimation,bickel1988estimating} and we use 
\begin{align}
\label{det105}
\hat {\cal M}_n^{(1)} :=  \frac{1}{n(n-1)h}\sum_{i \neq j}K\left(\frac{U_i-U_j}{h}\right)
\end{align}
 as estimate for the integrated squared density of $U$, where $h$ denotes a bandwidth, $K$ is a symmetric kernel function supported on $(-1,1)$ with integral $1$, for example the Epanechnikov kernel. The estimation of the first integral 
in \eqref{eq:optdist} is more intricate as it corresponds to the (squared) $L^2$-norm of a density on  $\mathbb{R}^+ \times \mathbb{S}^{p-1}$. For this purpose, we use classical results from density estimation on the sphere, such as 
\cite{HallWatsonCabrera1987},  and consider  a rapidly varying function $L: \mathbb{R}^p \to \mathbb{R}$, which satisfies (i) $L$ is nonnegative and nondecreasing, (ii) for each $0 < r < 1$,
\begin{align}
    L(rt) / L(t) \to 0,\ \text{as} \  t \to \infty.
\end{align}
The estimator of the first integral in Proposition \ref{mindist} is then defined by 
\begin{align}
\label{det106}
\hat {\cal M}_n^{(2)}:= \frac{1}{n(n-1) c_1(\kappa)h} \sum_{i \neq j} K\left(\frac{U_i-U_j}{h}\right) L(\kappa V_i^{\top} V_j),
\end{align}
where $K(\cdot)$ and  $h$ are the same kernel and bandwidth  as used in $\hat {\cal M}_n^{(1)}$, $\kappa \to \infty $ is another bandwidth corresponding to the spherical part of the vector $(U,V)$, and  
the constant $c_1(\kappa )$ is given by  
\begin{align}
    c_1(\kappa) = \omega_{p-2} \int_0^{\pi} L(\kappa \cos \theta) (\sin \theta)^{p-2} \mathrm d \theta .
    \label{def:c1k}
\end{align}
Throughout this paper, we will consider the Fisher-von-Mises distribution function    
    \begin{align}
    L(\kappa t) = \kappa^{p/2-1} \{ (2\pi)^{p/2} \mathcal I_{p/2-1}(\kappa)\}^{-1} e^{\kappa t}\label{eq:langevin}
    \end{align}
    for the kernel function $L(\cdot)$, 
    where  $\mathcal I_{\nu }$ is the modified Bessel function of the first kind  \cite[see also 9.6.18 on Page 376 of][]{abramowitz1968handbook} of order $\nu$, i.e.,
\begin{align}
   \mathcal I_\nu(\kappa) = \frac{(\kappa/2)^\nu}{\Gamma(\nu+1/2) \Gamma(1/2)}\int_{-1}^1 e^{\kappa t} (1 - t^2)^{\nu - 1/2} \dd t. \label{eq:bessel}
\end{align}
    
\noindent
Finally, we propose to estimate the minimum distance in \eqref{eq:optdist}  by the  $U$-statistic
\begin{align}
   \hat{\mathcal{M}}^2_n &= \hat{\mathcal{M}}^{(2)}_n - \omega_{p-1}^{-1} \hat{\mathcal{M}}^{(1)}_n   
  = \frac{2}{n(n-1)} \sum_{i=1}^n \sum_{j=1}^{i-1}H_n(Y_i, Y_j), \label{eq:Ustat}
\end{align}
where  $\hat{\mathcal{M}}^{(1)}_n$ and $\hat{\mathcal{M}}^{(2)}_n$ are defined in \eqref{det105} and \eqref{det106}, respectively,  and 
\begin{align}
    \label{det206}
H_n(Y_i, Y_j) = c_1^{-1}(\kappa)h^{-1} K\left(\frac{U_i - U_j}{h}\right) L(\kappa V_i^{\top} V_j )- \omega_{p-1}^{-1} h^{-1}  K\left(\frac{U_i - U_j}{h}\right)
\end{align}
is a kernel of order $2$, 
where $U_i = \|Y_i\|$, $V_i =Y_i / \|Y_i\| $ ($i=1, \ldots ,n)$. Note that the kernel depends on the sample size $n$ through the bandwidths $h$ and $\kappa$.

\section{Asymptotic properties and first statistical applications}
\label{sec3}
   \def\theequation{3.\arabic{equation}}	
   \setcounter{equation}{0}

   In this section, we establish a central limit theorem for the statistic $\hat{\mathcal{M}}_n^2$ as defined in \eqref{eq:Ustat}. Interestingly, the convergence rate of $\hat{\mathcal{M}}^2_n$ to $\mathcal{M}^2$ depends on whether the distribution of $Y$ is spherically symmetric or not.  We use these results to develop several  statistical applications in Section \ref{sec32}.

   \subsection{Asymptotic theory} \label{sec31}
 
\noindent
For the investigation of the bias of $\hat{\mathcal{M}}_n^2$, we recall a Taylor expansion for a twice continuously differentiable function $ f: \mathbb{R}^+ \times \mathbb{S}^{p-1} \to \mathbb{R}$  with respect to the argument $v$ on the sphere $\mathbb{S}^{p-1}$. According to equation  (3) in  \cite{DiMarzio03042014},  we have 
\begin{align}
    f(u, y) = f(u, x) + \theta \xi^{\T} \mathcal{D}_{f}(u, x) + \frac{\theta^2}{2} \xi^{\T} \mathcal{D}^2_{f}(u, x) \xi + O(\theta^3), \label{eq:taylorv}
\end{align}
where  $\mathcal D_{f}^s(u, x)$ is the matrix of the $s$th-order  derivative of the function $f(u, v)$ with respect to  $v$ at the point $x$ and for  $x,y \in \mathbb{S}^{p-1}$,  the vector $\xi$   and the angle $\theta \in (0, \pi)$ are defined by  the tangent–normal decomposition  
\begin{equation} \label{eq:tangent-normal}
    y = x \cos(\theta) + \xi \sin(\theta) .
\end{equation}
Note that  $\theta \in (0, \pi)$ is the angle between $x$ and $y$, and $\xi \in \mathbb{S}^{p-1}$ is a vector orthogonal to $x$, i.e., $\xi \in \Omega_x = \{z\in \mathbb{S}^{p-1}:  z \perp x\}$. 
The intuition for \eqref{eq:taylorv}  comes from the fact that $y-x \approx \theta \xi$ when $\theta$ approximates  $0$, which follows from a Taylor expansion of the function $\cos \theta $ and $\sin \theta $ at the point $\theta =0$.  
Moreover, the transformation of the area element of $\mathbb{S}^{p-1}$ corresponding to 
the tangent–normal decomposition 
\eqref{eq:tangent-normal} is given by 
    \begin{align}
         \omega_{p-1}(\dd y) = (\sin \theta)^{p-2} \dd \theta \ \omega_{p-2}(\dd \xi),
         \label{lm:measure}
    \end{align}
see equation (2) in  \cite{DiMarzio03042014} or  equation  (1.5)  in \cite{HallWatsonCabrera1987}.
Meanwhile, for a fixed $v \in \mathbb{S}^{p-1}$, we have the Taylor expansion for $f$ 
\begin{align}
    f(u_0+h, v) =  f(u_0, v) + h \left.\frac{\partial f(u, v) }{\partial u} \right|_{u = u_0} + \frac{h^2}{2} \left.\frac{\partial^2 f(u, v) }{\partial u^2} \right|_{u = u_0} + O(h^3),\label{eq:tayloru}
\end{align}
where $ \partial^s f(u, v) /\partial u^s$ denotes the partial  derivative of order $s$ with respect  to $u$.

Finally, we denote first and second derivatives of the marginal density $f_U$ by $f_U^{\prime}$ and $f_U^{\prime \prime}$, respectively.

 For a precise statement of our results we require the following assumptions. 

   \begin{assumption}
    $K(\cdot)$  is a symmetric  function supported on the interval $ (-1,1)$, $ \int_{-1}^1 K\left(u\right)=1 $ and $ \int_{-1}^1 K^4\left(u\right) \dd u < \infty$; the kernel $L(\cdot)$ is the Langevin kernel defined in 
    \eqref{eq:langevin}.

    \label{ass:kernel}
\end{assumption}

\begin{assumption}
 The density $f$ of $Y$ is twice continuously differentiable.  
\begin{itemize}
    \item[(a)] For any $u \in \R^+$, $\mathcal D_f^2(u, v)$ is uniformly continuous on $\mathbb{S}^{p-1}$. 
\item[(b)] For any $v \in \mathbb{S}^{p-1}$, $ \partial^2 f(u, v)/\partial u^2$ is uniformly continuous on $\R^+$.
\end{itemize}
\label{ass:density}
\end{assumption}

\begin{proposition}
\label{prop:meanvar}
Suppose that Assumption \ref{ass:kernel} and \ref{ass:density} are satisfied and that  $\kappa \to \infty$, $h \to 0$. 
\begin{itemize}
    \item [(i)] If the distribution of $Y$ is spherically symmetric, that is $\mathcal{M}^2 =0$, we have 
\begin{align}
    \E(\hat{\mathcal{M}}^2_n)= 0.
\end{align}
\item[(ii)] If the distribution of $Y$ is not spherically symmetric, that is $\mathcal{M}^2 > 0$, we have  
\begin{align}
       \E(\hat{\mathcal{M}}^2_n) 
       &= {\mathcal{M}}^2  \\
       & +\frac{h^2 \phi_2(K) }{2} \left\{\int_{\mathbb{S}^{p-1}} \int_{\mathbb R^+} f(u_2,v_1)\left. \frac{\partial^2 f(u, v_1)}{\partial u^2} \right|_{u=u_2} \mathrm du_2 \ \omega_{p-1}(\mathrm  dv_1)  - \omega_{p-1}^{-1} \int_{\mathbb R^+} f_U(u) f_U^{\prime \prime}(u)\mathrm du \right\}\\ 
       &+ \frac{\omega_{p-2}}{2\kappa}\int_{\mathbb{S}^{p-1}} \int_{\mathbb R^+} f(u_2, v_1)   \mathrm{tr}\{ D^2_{f}(u_2, v_1)\}      \dd u_2 \omega_{p-1}(\dd v_1)+ O(\kappa^{-3/2}+ h^3 + h^2\kappa^{-1}), 
\end{align}
where  for $j\geq 1$
\begin{align}
    \label{det210}
\phi_j(K)  = \int_{-1}^1 x^j K(u) du. 
\end{align}
\end{itemize}
\end{proposition}

\noindent
By Proposition \ref{prop:meanvar}, the statistic  $\hat {\mathcal M}^2_n$ is an  unbiased estimator of $\mathcal{M}^2$ only under  spherical symmetry. When the exact spherical symmetry does not hold,  the bias of the kernel estimation with respect to the radius $\| Y\| $ is of order $h^2$, and it is of order $1/\kappa$ with respect to  the direction $Y/ \| Y\| $.    
   Therefore, smaller values of  $h$ and larger values of  $\kappa$ are  preferred to diminish the bias  of the estimator $\hat{\mathcal{M}}_n^2$. We now turn to the weak convergence of $\hat{\mathcal{M}}_n^2$.  Throughout this paper, the symbol $\stackrel{d}{\longrightarrow} $ denotes weak convergence (convergence in distribution) of real-valued random variables.

\begin{theorem}
     If  \cref{ass:kernel} and \ref{ass:density} are satisfied, $\kappa \to \infty$, $h \to 0$, $\kappa^{(p-1)/2}/(nh)  \to 0$, and the distribution of $Y$ is spherically symmetric, that is $\mathcal{M}^2=0$, we have 
     $$
     {\hat{\mathcal{M}}^2_n  \over s_n} \stackrel{d}{\longrightarrow} 
     {\cal N} (0,1),
     $$
     where
     \begin{align} \label{det108}
         s_n^2 =  
         \frac{\psi_2(K)    \kappa^{(p-1)/2} \int_{\mathbb{S}^{p-1}}\int_{\mathbb R^+} f^2(u,v) \mathrm du \ \omega_{p-1}(\mathrm  dv)}{2^{p-2} \pi^{(p-1)/2}n(n-1)h
       },
     \end{align}
     and 
     \begin{align}
         \psi_j(K)  = \int_{-1}^1 K^j(u) du. 
     \end{align}  \label{thm:normal}
\end{theorem}

 \noindent
 Note that the convergence rate in \cref{thm:normal} is of order  $\sqrt{n(n-1)h/{\kappa^{(p-1)/2}}}$, which is asymptotically greater than $
\sqrt{n}$ as $\kappa^{(p-1)/2}/(nh)  \to 0$, by assumption. On the other hand, weak convergence with a rate $
\sqrt{n}$ is obtained in the case where the distribution of $Y$ is not spherically symmetric.

\begin{theorem}
Assume that   \cref{ass:kernel} and \ref{ass:density} are satisfied and that 
$\kappa \to \infty$, $h \to 0$, $\kappa^{(p-1)/2}/(nh)  \to 0$. If the distribution of $Y$ is not spherically symmetric, that is $\mathcal{M}^2 > 0$, we have 
    $$
    \sqrt{n}\big (\hat{\mathcal{M}}^2_n -\mathbb{E} [\hat{\mathcal{M}}^2_n ] \big ) \stackrel{d}{\longrightarrow}  {\cal N } (0, 4 \sigma^2) ,
    $$
 where the asymptotic variance is given by 
\begin{align}
 \sigma^2 & = \int_{\R^+} \int_{ \mathbb{S}^{p-1}} f (u,v)\big (f (u,v) - \omega_{p-1}^{-1} f_U (u) \big )^2\dd u \ \omega_{p-1}(\mathrm  dv)  - ({\mathcal{M}}^2)^2. 
  \label{det109}
   \end{align}
 \label{thm:alter}
\end{theorem}

\begin{remark} \label{rembias}  {\rm 
Under exact sphericity (${\cal M}^2=0$), it follows from Theorem \ref{thm:normal}  that the variance  of ${\cal M}^2_n$ is of order $\kappa^{(p-1)/2} / (n(n-1)h)$  and the estimator is consistent independently of the dimension. 

In the case 
 ${\cal M}^2> 0$, the situation is more complicated. 
 Theorem \ref{thm:alter} can be used to derive the weak convergence  
\begin{align}
    \label{bias0}
    \sqrt{n}(\hat{\mathcal{M}}^2_n -{\mathcal{M}}^2) \stackrel{d}{\longrightarrow}  {\cal N } (0, 4 \sigma^2).
\end{align}
    However, some care is necessary when replacing 
    $\mathbb{E} [\hat{\mathcal{M}}^2_n ]$ by $\mathcal{M}^2$ as  the error  of this replacement has to be of  order $o(1/\sqrt{n})$.  By Proposition \ref{prop:meanvar}(ii) this yields to the condition 
  \begin{align}
  \label{bias1}    
 \sqrt{n}/\kappa +\sqrt{n}h^2 \to 0, 
  \end{align}
which  
 has to be simultaneously satisfied with the condition $\kappa^{(p-1)/2}/(nh)  \to 0$  required for controlling the variance of the $U$-statistic. This   yields some restriction on the dimension $p$.
For example, if  $p=3$, the condition \eqref{bias1} can be satisfied by setting $h = n^{-3/8}$, $\kappa = \lfloor n^{9/16} \rfloor$. However,  a simple calculation shows that \eqref{bias1} and  $\kappa^{(p-1)/2}/(nh)  \to 0$ cannot  simultaneously hold  if $p \geq 4$. 

This issue can be mitigated via a simple bias reduction 
\begin{align}
    \tilde{\mathcal{M}}^2_n(a, \kappa) = \frac{1}{1-a}  \hat{\mathcal{M}}^2_n(\kappa) - \frac{a}{1-a}  \hat{\mathcal{M}}^2_n(a\kappa),
\end{align}
where $\hat{\mathcal{M}}^2_n(a\kappa)$ denotes the estimator $\hat{\mathcal{M}}^2_n$ using $a\kappa$ instead of $\kappa$  in the rapidly varying function $L(\cdot)$ in \eqref{eq:langevin}, see  also Theorem 3 of \cite{TSURUTA2024105338}. In addition, if we replace the kernel $K(\cdot)$ by its   Jackknife correction $\tilde K (x)  = 2\sqrt{2} K(\sqrt{2}x) - K(x)$, which satisfies $\phi_2(\tilde K) = 0$,  we obtain from   Proposition \ref{prop:meanvar}(ii)
\begin{align}     \E(\hat{\mathcal{M}}^2_n) 
       &= {\mathcal{M}}^2 + O(\kappa^{-3/2}+ h^3),
       \label{det200a}
\end{align}
and  \eqref{bias0} holds under the conditions 
$\sqrt{n}/\kappa^{3/2} +\sqrt{n}h^3 \to 0$ and $\kappa^{(p-1)/2}/(nh)  \to 0$, which can be satisfied, whenever $p \leq 6$, for example $\kappa = n^{\frac{3}{p+3}}, h = n^{-\frac{3}{2(p+3)}}$.  If  $p >6 $, corresponding results can be derived  under additional smoothness assumptions  using the higher order expansions 
\begin{align}
    f(u, y) & = f(u, x) + \sum_{s=1}^q \frac{\theta^s}{s!} \xi^{\T} \mathcal{D}^s_{f}(u, x) \xi^{\otimes (s-1)} + O(\theta^{q+1}),  \\
      f(u_0+h, v) & =  f(u_0, v) + \sum_{s=1}^q \frac{h^s}{s!} \left.\frac{\partial^s f(u, v) }{\partial u^s} \right|_{u = u_0} + O(h^{q+1}) , 
\end{align}
where $u^{\otimes s}$ denotes the $s$-th Kroneckerian power of the vector $u$. For example,  if $q=3$, it follows, observing   $\int_{\Omega_x} \xi^{\top} \xi^{\otimes2} \omega_{p-2}(\xi) = 0$  and  $\int x^3 K(x) d x = 0$ ($K$ is symmetric by Assumption \ref{ass:kernel}), that  the rate in \eqref{det200a} can be improved to $ O( \kappa^{-2} + h^4 ) $, and \eqref{bias0} holds under the conditions 
\begin{align}
    \sqrt{n}/\kappa^{2} +\sqrt{n}h^4 \to 0, \label{eq:rmassumption}
\end{align} which can be satisfied, whenever $p \leq 8$, for example $\kappa = n^{\frac{3}{p+4}}, h = n^{-\frac{3}{2(p+4)}}$.
}

\label{rm:bias}
\end{remark}

 \subsection{Some statistical applications} \label{sec32} 

 By Proposition \ref{prop:meanvar},  Theorem \ref{thm:normal} and \ref{thm:alter} and Remark \ref{rm:bias}, the statistic $\hat{\mathcal{M}}_n^2$ is a consistent estimator of the measure $\mathcal{M}^2$, which defines the deviation from sphericity. To use these results for uncertainty quantification, one requires estimates of the variances $s_n^2$ and $\sigma^2$ defined in \eqref{det108}
 and \eqref{det109}, respectively.

A simple estimator for  $s_n^2$ is given by
\begin{align}
    \hat s_n^{2} = \frac{\phi_2(K)  \kappa^{(p-1)/2} }{2^{p-2} \pi^{(p-1)/2} n^2(n-1)^2 c_1(\kappa)h^2} \sum_{i \neq j} K\left(\frac{U_i-U_j}{h}\right) L(\kappa V_i^{\top} V_j),
    \label{det209}
\end{align}
which is obtained by  replacing the integral of the squared density in \eqref{det108} by its corresponding estimate \eqref{det106}.  We will use this estimate to define a test for exact sphericity (see Remark \ref{rem31} below).

The  estimation of $\sigma^2$ is more difficult and we propose a  Jackknife approach for  this purpose \citep[see Chapter 5 of][]{Lee2003}).
To be precise, let $\hat {\mathcal{M}}^2_{n-1}(-i)$ denote the estimator 
\eqref{eq:Ustat}
of  the minimum distance ${\cal M}^2$ based on the observations $Y_1, \ldots, Y_{i-1}, Y_{i+1},\ldots, Y_n$, and define the pseudovalues 
\begin{align}
    \tilde {\mathcal{M}}^2_i  = n \hat {\mathcal{M}_n}^2 - (n-1) \hat {\mathcal{M}}^2_{n-1}(-i)
\end{align}
($i=1, \ldots , n)$. 
The Jackknife estimator of the asymptotic variance $\sigma^2 = \lim_{n\to \infty} {\rm Var} (\sqrt{n} \mathcal{M}_n^2 )$ is then given by 
\begin{align}
 \hat \sigma^2_n  =  \frac{1}{4(n-1)} \sum_{i=1}^n (\tilde {\mathcal{M}}_i^2 - \bar {\mathcal{M}}_n^2)^2,
\end{align}
where  
\begin{align}
  \bar {\mathcal{M}}_n^2 =  n^{-1} \sum_{i=1}^n \tilde {\mathcal{M}}_i^2
\end{align}
denotes the Jackknife estimate 
 of the mean  of $\tilde {\mathcal{M}}_i^2$. 

From Theorem \ref{thm:alter} and Remark \ref{rm:bias}, we obtain a simple asymptotic confidence  interval for $\mathcal{M}^2$, that is 
\begin{equation}
    \label{det110}
    \hat I_n = \Big [ \hat{\mathcal{M}}_n^2 -  {\hat \sigma_n \over \sqrt{n}}u_{1-\alpha/2}  , \hat{\mathcal{M}}_n^2 +   {\hat \sigma_n \over \sqrt{n}} u_{1-\alpha/2} \Big  ], \label{eq:jackI}
\end{equation}
where  $u_{1-\alpha/2}$ denotes the $(1-\alpha/2)$-quantile of the standard normal distribution.  The following result, which follows from Theorem \ref{thm:alter} and the consistency of $\hat \sigma_n^2$ for $\sigma^2$, shows that this interval keeps its nominal level asymptotically. 
 
\begin{corollary}
\label{cor31}
 If the assumptions in  Theorem \ref{thm:alter} and \eqref{bias1} in Remark \ref{rembias} are satisfied, we have
$$
\lim_{n\to \infty}  \mathbb{P} \big ( \mathcal{M}^2 \in \hat I_n \big ) = 1- \alpha .
$$
\end{corollary}

\noindent
Next, we turn to the problem of testing for spherical  symmetry. As pointed out in the introduction  we are not interested in testing for exact sphericity, that $\mathcal{M}^2=0$, because there are many applications where one does not really believe in exact sphericity, but wants to assume this with the hope that the deviations from sphericity are small and a  procedure developed under the assumption of exact sphericity (such as classical ANOVA) still yields reliable and efficient inference. With this point of view, we propose to test the hypotheses 
\begin{align}
    \label{det111}
    H_0^{\rm rel}:  \mathcal{M}^2 \leq \Delta ~~~~{ \rm versus }  ~~~~ H_1^{\rm rel}:  \mathcal{M}^2 > \Delta ~, \\
      \label{det112}
    H_0^{\rm eq}:  \mathcal{M}^2 \geq   \Delta ~~~~{ \rm versus }  ~~~~ H_1^{\rm eq}:  \mathcal{M}^2 <  \Delta ~, 
\end{align}
where $\Delta>0$ is a prespecified threshold.
Note that this perspective of hypothesis testing aligns with the view expressed by \cite{berger1987}, who argue that it is {\it rare, and perhaps impossible, to have a null hypothesis that can be exactly modeled by a parameter being precisely $0$}. Similarly, \cite{tukey1991}, in the context of multiple comparisons of means, emphasizes that {\it ``All we know about the world teaches us that the effects of A and B are always different — in some decimal place — for any A and B. Thus, asking ‘Are the effects different?’ is foolish.''}  We call hypotheses the form \eqref{det111} and \eqref{det112} relevant hypotheses in the following discussion. Note also that hypotheses of the form  \eqref{det112} have found considerable attention in {\it equivalence testing} in the field of biostatistics, which explains the notations $H_0^{\rm eq}$ and $H_1^{\rm eq}$   in \eqref{det112} \citep[see][]{wellek2010testing} and that rejection of the null hypothesis in \eqref{det112} allows to decide for approximate sphericity at a controlled type I error.

We propose to reject the null hypothesis in \eqref{det111} if 
\begin{align}
    \label{test1}
 \hat{ \mathcal{M}}_n^2 > \Delta + u_{1-\alpha} {\hat \sigma_n  \over \sqrt{n} } .
\end{align}
Similarly, we propose to reject the null hypothesis in \eqref{det112}, whenever
\begin{align}
    \label{test2}
 \hat{ \mathcal{M}}_n^2 \leq  \Delta + u_{\alpha} {\hat\sigma_n  \over \sqrt{n} } .
\end{align}
Our next results shows that both decision rules define consistent and asymptotic level $\alpha$ tests for the hypotheses \eqref{det111} and \eqref{det112}, respectively.

\begin{corollary}
    \label{cor2}
     ~~
     Let the assumptions of Theorem \ref{thm:alter} be satisfied.
    \begin{itemize}
        \item [(a)] For the  test \eqref{test1}, we have 
\begin{align}
\lim_{n\rightarrow\infty}
\mathbb{P }\Big(\hat{\mathcal{M}}_n^2 >  \Delta + u_{1- \alpha} {\hat \sigma_n \over \sqrt{n} }  \Big) =     
\begin{cases}
1,      & \text{if } {\mathcal{M}}^2 >  \Delta,\\ 
\alpha, & \text{if } {\mathcal{M}}^2 = \Delta  ,\\
0,      & \text{if } {\mathcal{M}}^2 < \Delta.
\end{cases}
\end{align}
        \item[(b)] For the test \eqref{test2}, we have 
\begin{align}
\lim_{n\rightarrow\infty}
\mathbb{P }\Big(\hat{\mathcal{M}}_n^2 \leq \Delta + u_{\alpha} {\hat \sigma_n \over \sqrt{n} }  \Big) =     
\begin{cases}
1,      & \text{if } {\mathcal{M}}^2 <  \Delta,\\ 
\alpha, & \text{if } {\mathcal{M}}^2 = \Delta  ,\\
0,      & \text{if } {\mathcal{M}}^2 > \Delta.
\end{cases}
\end{align}
           \end{itemize}
\end{corollary}

\noindent
We will illustrate the finite sample properties of the tests \eqref{test1} and \eqref{test2} in Section \ref{sec5}.

\begin{remark}~~
\label{rem31}
   {\rm 
   \begin{itemize}
       \item[(a)]
       An important question in testing hypotheses of the form \eqref{det111} and \eqref{det112} is the choice of the threshold $\Delta $. In the following, we shall explain how to choose it in a data adaptive way for the hypotheses \eqref{det112}, and a similar argument can be given for the hypotheses \eqref{det111}. Note that these hypotheses are nested, that is, if $H_0^{\rm eq}$ is satisfied for $\Delta_1$, then it is also satisfied for all $\Delta_2 \leq \Delta_1$. The decision rules for different $\Delta$'s are nested in the same way(if $H_0^{\rm eq}$ is rejected for $\Delta_2$, it is also rejected for all $\Delta_1 \geq \Delta_2$). 
  By the sequential rejection principle, we may simultaneously test the  hypotheses  in \eqref{det112}    for different $\Delta \geq 0$ 
  starting at $\Delta  = 0$ and 
   increasing  $\Delta $ to 
   find the minimum value of $\Delta $, say
  \begin{align}
     \label{det113}
 \hat \Delta_\alpha:= \min 
 \big \{\Delta \ge 0 \,| \,  \hat {\cal M}^2_n \leq  \Delta + u_{\alpha} \hat \sigma_n /\sqrt{n} \big  \}  
 \end{align}
   for which 
   $H_0$  in \eqref{det112} is  rejected.
  The quantity  $\hat \Delta_\alpha $ could be interpreted as a measure of evidence against the null hypothesis in \eqref{det112}.
  In this sense, the question of a reasonable
choice of the threshold $\Delta$ can  be postponed until after seeing the data. 
       \item[(b)] Although the focus of this paper is not on the classical hypothesis of exact sphericity, that  this 
\begin{align}
    \label{det208}
  H_0^{\rm exact }:   \mathcal M^2 =0 ~~~~ \text{versus} ~~~~\mathcal H_1^{\rm exact }: \mathcal M^2 >0  ~,
\end{align}
   it is worthwhile mentioning that  the theory developed so far also provides a test for these hypotheses without extra restrictions on the dimension $p$. More precisely, we propose to reject the null hypothesis  in \eqref{det208},  whenever
\begin{align}
    \label{testexact}
 \hat{ \mathcal{M}}_n^2 >
 u_{1-\alpha} {\hat s_n  \over \sqrt{n} } ,
\end{align}
where $\hat  s_n$ is defined in \eqref{det209}. Then, it follows from \cref{thm:normal} and \ref{thm:alter} that this decision rule defines a consistent and asymptotic level $\alpha$ test. In Section \ref{sec:compare}, we compare the finite sample performance of the test \eqref{testexact} for the hypotheses \eqref{det208} with other existing methods. 
    \end{itemize}
    }
\end{remark}

\section{Pivotal inference} \label{sec4}
  \def\theequation{4.\arabic{equation}}	
   \setcounter{equation}{0}

The statistical methodology  developed in Section \ref{sec32} requires the estimation of 
the variance $\hat s_n^2$.  In this section, we develop a pivotal confidence interval for  the deviation $\mathcal M^2$ from sphericity and  pivotal tests for the hypotheses \eqref{det111}   and \eqref{det112}.  For this purpose, we prove a weak convergence result for a sequential version of the estimator $\hat{\mathcal M}_n^2$.

To be precise, let $\hat{\mathcal M}^2_{\lfloor nt \rfloor}$
denote the estimator 
\eqref{eq:Ustat} calculated for the sample 
$Y_1, \ldots , Y_{\lfloor nt \rfloor}$, where $t \in [0,1]$. We consider the sequential process
\begin{align}
\label{det207a}
    S_n(t) = \frac{\lfloor nt \rfloor}{2\sqrt{n}} \big (\hat{\mathcal M}^2_{\lfloor nt \rfloor} - \E [\hat{\mathcal M}^2_n] \big ), 
\end{align}
whose asymptotic properties are investigated in the following theorem.  
Throughout this section, the symbol $\Rightarrow$ denotes weak convergence 
in the space $\ell^\infty ([0,1])$ of bounded functions on the interval $[0,1]$.
\begin{theorem}
\label{thm:invariance}
If  the conditions of \cref{thm:alter} are satisfied  and $(\log n)^2 \kappa^{(p-1)/2}/(nh) \to 0$,  we have 
\begin{align}
 \big \{ S_n(t) \big \}_{t \in [0,1]} \Rightarrow \big \{ \sigma \mathbb{B} (t) \big \}_{t \in [0,1]},
\end{align}
where $\big \{ \mathbb{B} (t) \big \}_{t \in [0,1]} $ is a standard Brownian motion and $\sigma^2 $ is defined in \eqref{det109}. 
\end{theorem}

\noindent 
\begin{remark}\label{rm:invariance}
  It follows from   Proposition \ref{prop:meanvar} that 
  $$
  \frac{1}{\sqrt{n}}\max_{2 \leq k \leq n} \Big |k \big  \{\mathcal M^2 - \E [ H_n(Y_1, Y_2)]  \big\} \Big | =\sqrt{n} \big | \mathcal M^2 - \E [ H_n(Y_1, Y_2) ]  \big | = O(\sqrt{n}/\kappa+ \sqrt{n}h^2)=o(1)
  $$
  provided that the bandwidth condition \eqref{bias1} is satisfied.
Along with \cref{thm:invariance}, we therefore have 
\begin{align}
    \frac{\lfloor nt \rfloor}{2\sqrt{n}} \big (\hat{\mathcal M}^2_{\lfloor nt \rfloor} - {\mathcal M}^2 )  \Rightarrow \big \{ \sigma \mathbb{B} (t) \big \}_{t \in [0,1]}.\label{eq:invarianceM}
\end{align}
The bandwidth condition \eqref{bias1}  can be relaxed as discussed in Remark \ref{rm:bias}.
\end{remark}
We will now use the result \eqref{eq:invarianceM} to develop pivotal inference for the measure $\mathcal{M}^2$. In particular, we define the statistic
\begin{align}
\label{det205}
\hat{V}_n = \int^1_0 \big|\hat{\mathcal M}^2_{\lfloor nt \rfloor}  -  \hat{\mathcal M}^2_n \big|\, t \,\dd t ,
\end{align}
and note that it follows from the continuous mapping theorem  and \cref{thm:invariance} that 
\begin{align}\label{def W}
    \frac{\hat{\mathcal M}^2_n -{\mathcal M}^2 }{\hat{V}_n}\;\stackrel{d}{\longrightarrow} \; W =  
    \frac{\mathbb{B}(1)}{\int^1_0|\,\mathbb{B}(t) -  t\mathbb{B}(1)|\dd t }. 
\end{align}

In the following, let $q_{1-\alpha }$ denote the $(1-\alpha)$ quantile of the distribution of $W$. Then, a pivotal confidence interval for 
$\mathcal M^2$ is given by  
\begin{equation}
    \label{det110a}\label{eq:pivI}
    \hat I_n^{\rm piv} = \Big [ \hat{\mathcal{M}}_n^2 -  \hat V_n q_{1-\alpha/2}  , \hat{\mathcal{M}}_n^2 +  \hat V_n q_{1-\alpha/2} \Big  ],
\end{equation}
Similarly, we propose tests for the hypotheses  \eqref{det111} and \eqref{det112} to reject the null hypothesis, whenever
\begin{align}
    \label{test1a}
 \hat{ \mathcal{M}}_n^2 & > \Delta + q_{1-\alpha} {\hat V_n  },
 \end{align}
 and 
 \begin{align}
    \label{test2a}
 \hat{ \mathcal{M}}_n^2 & \leq  \Delta + q_{\alpha} {\hat V_n },
\end{align}
respectively. We note that the comments regarding the choice of the threshold $\Delta$ in Remark \ref{rem31}(a) remain valid for the pivotal tests \eqref{test1a} and \eqref{test2a},  and summarize the properties of these tests in the following corollary.

\begin{corollary}
\label{cor41}
 If the assumptions in  Theorem \ref{thm:invariance} and in Remark \ref{rm:invariance} are satisfied, we have
$$
\lim_{n\to \infty}  \mathbb{P} \big ( \mathcal{M}^2 \in \hat I_n^{\rm piv} \big ) = 1- \alpha .
$$
Moreover, the tests \eqref{test1a} and \eqref{test2a} are consistent and have asymptotic level $\alpha$ for the hypotheses \eqref{det111} and \eqref{det112}, respectively.
\end{corollary}

\section{Finite sample properties} \label{sec5}
  \def\theequation{5.\arabic{equation}}	
   \setcounter{equation}{0}

In this section, we illustrate the finite sample properties of the developed methodology through a small simulation study. For the choice of kernel functions, we use the Fisher-von-Mises distribution function\eqref{eq:langevin} for $L(\cdot)$ and Epanechnikov kernel for $K(\cdot)$, i.e., $K(x) = \frac{3}{4} (1 - x^2)$, for $|x|< 1$, and $K(x) = 0$ otherwise. Both of the functions satisfy the Assumption \ref{ass:kernel}.  In addition, we implement the bias correction procedure in 
Remark \ref{rm:bias} with $a = 0.9$. All results presented here are based on $1000$ simulation runs.
\medskip

According to \cref{thm:alter}, the term $\kappa^{(p-1)/2}/(nh)$ controls the variance for the non-leading term in the Gaussian approximation, while $\sqrt{n}/\kappa^2 + \sqrt{n}h^4$ accounts for the bias term with  the bias reduction procedure discussed in \cref{rem31}. In addition, for the sequential convergence, by \cref{thm:invariance}, the inflated rate  $(\log n)^2 \kappa^{(p-1)/2}/(nh)$ controls the variance of the process of the non-leading term. In order to approximate Gaussianity in finite samples better, we recommend imposing a stronger control over the non-leading term, that is,
$n^{3/2} \times \kappa^{(p-1)/2}/(nh)$.
Therefore, we consider bandwidths $(h,\kappa)$ satisfying 
\begin{align}
    (\sqrt{n}/\kappa^2 + \sqrt{n}h^4)^2 \asymp \kappa^{(p-1)/2} \sqrt{n}/h. \label{eq:bdreduce}
\end{align}
For a data-dependent choice of the  smoothing parameters, we first define a sequence for  $(h, \kappa)$'s, i.e., $(h_i, \kappa_i):= (n^{-1/(2(p+8))}a_i, n^{1/(p+8)}c_i)$ satisfying \eqref{eq:bdreduce}, where $a_i$ and $c_i$ are prespecified constants, $i = 1,\ldots, M$. The choices of $a_i$ and $c_i$ can be obtained by inspecting the turning point of the graph of $\hat{\mathcal M}_n^2$ versus $h$ and that versus $\kappa$. We give more details for each model below. 

Given the chosen grid, we select the pair of parameters that minimizes the volatility of the self-normalizing  term $\hat V_n$ in \eqref{det205}, since its expectation is proportional to the standard deviation of the statistic $\hat {\cal M}_n^2$. More precisely, we select $(h, \kappa)$ as the minimizer of  
\begin{align}
   \min_{i \in \{1, \ldots, M\}}   \mathrm{SE}\{\hat V_n(h_j, \kappa_j)\}_{j=i-1}^{i+1},\label{eq:SE}
\end{align}
where SE denotes the standard error.
\medskip

\noindent
In the following discussion, we consider two models. 
\medskip

\noindent
\textbf{Model 1:} $Y$  has a $3$-dimensional Gaussian distribution with mean vector $\mu = (1,0,2)^{\top}$ and covariance matrix
\begin{align}
    \Sigma = 0.25 \times \begin{pmatrix}
        1& 0.3 &0 \\ 
        0.3&1&0\\ 
        0&0&1
    \end{pmatrix}.
\end{align}
Since it is difficult to obtain the analytical solution of $\mathcal M^2$, we calculate the deviation $\mathcal M^2$  numerically. We obtain $\mathcal M^2 \approx 0.95$,
 using a larger sample size $2500$ and averaged over $20000$ times of repeated simulation. For the selection of the  smoothing parameters via \eqref{eq:SE}, we use 
$$
(c_i)_{i=1}^5 = (72.50, 73.75, 75.00, 76.25, 77.50)^{\top}~,~~ (a_i)_{i=1}^5 = ( 0.7500, 0.8125, 0.8750, 0.9375, 1.0000)^{\top}.
$$

\medskip

\noindent
\textbf{Model 2:} $Y$ has a $5$-dimensional  Gaussian distribution with mean vector $\mu = (1, 0,0,  -2 ,0)^{\top}$ and covariance matrix
\begin{align}
    \Sigma = 0.25 \times \begin{pmatrix}
        1& 0.2 &0 &0 &0 \\ 
        0.2&1&0.3 &0 &0 \\ 
        0&0.3&1&0 &0 \\
        0 & 0 &0 & 1 & 0.2\\ 
         0 & 0 &0 & 0.2 & 1 
    \end{pmatrix}.
\end{align}
The deviation $\mathcal M^2$  was calculated numerically as $\mathcal M^2 \approx  1.97$.
For the selection of the  smoothing parameters via \eqref{eq:SE}, we use 
$$
(c_i)_{i=1}^5 =(37.50, 38.75, 40.00, 41.25, 42.50)^{\top} ~,~~  (a_i)_{i=1}^5 = ( 0.7500, 0.8125, 0.8750, 0.9375, 1.0000)^{\top}.
$$

\subsection{Confidence intervals for the deviation from a spherical distribution}\label{sec:ecover}

In \cref{tb:cover1} and \ref{tb:cover2}, we display  the empirical coverage rates  and the average widths of Jackknife confidence intervals  \eqref{eq:jackI}  and   pivotal  confidence intervals 
\eqref{eq:pivI} for the minimum distance \eqref{det103}, where we choose the nominal levels as  $95\%$ and $90\%$.

\begin{table}[t]
\centering
\begin{tabular}{rrrrrrrrr}
  \hline
 & \multicolumn{4}{c}{Jackknife}& \multicolumn{4}{c}{Pivotal} \\ 
   & \multicolumn{2}{c}{coverage} & \multicolumn{2}{c}{width} & \multicolumn{2}{c}{coverage} & \multicolumn{2}{c}{width} \\ 
$n$ & 95\% & 90\% & 95\% & 90\% &95\% & 90\% &95\% & 90\% \\ 
  \hline
 200 & 95.9 & 91.7 & 0.46 & 0.39 & 97.8 & 95.3 & 0.71 & 0.56 \\ 
  300 & 96.5 & 91.7 & 0.36 & 0.30 & 97.4 & 93.7 & 0.53 & 0.42 \\ 
  400 & 94.6 & 89.6 & 0.31 & 0.26 & 96.7 & 92.0 & 0.45 & 0.35 \\ 
  500 & 95.4 & 91.0 & 0.27 & 0.23 & 97.0 & 92.5 & 0.39 & 0.31 \\ 
  600 & 94.9 & 89.8 & 0.25 & 0.21 & 96.5 & 91.2 & 0.36 & 0.28 \\ 
  800 & 94.5 & 90.8 & 0.21 & 0.18 & 95.2 & 90.4 & 0.30 & 0.23 \\ 
  1000 & 95.4 & 89.1 & 0.19 & 0.16 & 95.2 & 89.8 & 0.25 & 0.20 \\ 
   \hline
\end{tabular}
\caption{\it Simulated coverage rates (in \%) and average widths  of  Jackknife and pivotal confidence intervals  for Model 1.}\label{tb:cover1}
\end{table}

\begin{table}[b]
\centering
\begin{tabular}{rrrrrrrrr}
  \hline
 & \multicolumn{4}{c}{Jackknife}& \multicolumn{4}{c}{Pivotal} \\ 
   & \multicolumn{2}{c}{coverage} & \multicolumn{2}{c}{width} & \multicolumn{2}{c}{coverage} & \multicolumn{2}{c}{width} \\ 
$n$ & 95\% & 90\% & 95\% & 90\% &95\% & 90\% &95\% & 90\% \\ 
  \hline
200 & 94.7 & 90.1 & 1.19 & 1.00 & 98.1 & 94.9 & 1.87 & 1.46 \\ 
  300 & 95.0 & 91.0 & 0.94 & 0.79 & 97.0 & 93.0 & 1.42 & 1.11 \\ 
  400 & 94.8 & 90.3 & 0.80 & 0.67 & 97.6 & 93.8 & 1.17 & 0.92 \\ 
  500 & 94.1 & 89.4 & 0.71 & 0.59 & 96.3 & 92.8 & 1.02 & 0.80 \\ 
  600 & 93.9 & 90.0 & 0.65 & 0.54 & 95.9 & 91.4 & 0.93 & 0.72 \\ 
  800 & 93.5 & 88.1 & 0.55 & 0.46 & 95.0 & 90.3 & 0.78 & 0.61 \\ 
  1000 & 94.3 & 89.4 & 0.49 & 0.41 & 94.7 & 90.4 & 0.69 & 0.54 \\ 
   \hline
\end{tabular}
\caption{\it Simulated coverage rates (in \%) and average widths  of  the Jackknife and pivotal  confidence intervals  for Model 2.}\label{tb:cover2}
\end{table}

For Model 1, the empirical coverage rates of both methods are close to their nominal levels for sample sizes larger than $500$. The Jackknife confidence interval achieves more accurate simulated coverage rates when the sample size is small, while the pivotal confidence intervals are more conservative for smaller sample sizes. A potential explanation for this observation is that when the sample size is small, the  widths of the pivotal method are much larger than those of the Jackknife method.  

We observe similar patterns for Model 2.  The pivotal confidence intervals are more conservative compared to the Jackknife confidence intervals but achieve accurate coverage rates for large sample sizes. At the same time, the Jackknife confidence intervals have narrower  widths and work very well even when the sample size is small.

\subsection{Relevant hypothesis}
  \label{sec52}

In \cref{tb:rel1} and \ref{tb:rel2},  we display the simulated rejection rates of Jackknife test  \eqref{test2}   and the pivotal  test \eqref{test2a} for the hypotheses \eqref{det112}  for 
different values of $\Delta$  and a significance level of $5\%$. For both models the results reflect the asymptotic properties of the test described  in  \cref{cor2}.
At the ``boundary'' of the hypotheses, where  $ {\cal M}^2=\Delta$,  the empirical sizes of
the pivotal test \eqref{test2a}  are close to the nominal level $5\%$ in most cases, while the type I error of the Jackknife test \eqref{test2} is too large for small sample sizes. \begin{table}[t]
\centering
\begin{tabular}{ccccccc|cccccc}
  \hline
 &\multicolumn{6}{c|}{$H_0$}& \multicolumn{6}{c}{$H_1$}\\ 
   \hline
  &\multicolumn{2}{c}{$\Delta = 0.7$}& \multicolumn{2}{c}{$\Delta = 0.9$}& \multicolumn{2}{c|}{$\Delta = \mathcal M^2$}&\multicolumn{2}{c}{$\Delta = 1.1$}& \multicolumn{2}{c}{$\Delta = 1.3$}& \multicolumn{2}{c}{$\Delta = 1.5$} \\
    \hline
$n$ & \eqref{test2}  & \eqref{test2a} & \eqref{test2}  & \eqref{test2a} & \eqref{test2}  & \eqref{test2a} & \eqref{test2}  & \eqref{test2a} & \eqref{test2} & \eqref{test2a} & \eqref{test2} & \eqref{test2a}\\ 
  \hline
200 &  0 &  0 & 2.0 & 1.4 & 7.7 & 3.9 & 39.2 & 23.1 & 87.8 & 68.2 & 99.1 & 91.3 \\ 
  300 &  0 &  0 & 1.3 & 0.8 & 6.1 & 4.2 & 51.9 & 36.2 & 95.4 & 82.4 & 100.0 & 98.3 \\ 
  400 &  0 &  0 & 0.9 & 0.8 & 5.8 & 4.9 & 60.3 & 43.0 & 98.7 & 91.5 & 100.0 & 99.7 \\ 
  500 &  0 &  0 & 0.5 & 0.6 & 6.2 & 4.5 & 71.2 & 53.7 & 99.6 & 95.9 & 100.0 & 100.0 \\ 
  600 &  0 &  0 & 0.5 & 0.7 & 5.1 & 3.9 & 74.3 & 59.6 & 99.8 & 97.6 & 100.0 & 99.9 \\ 
  800 &  0 &  0 & 0.5 & 0.8 & 3.9 & 3.9 & 83.0 & 65.6 & 100.0 & 99.3 & 100.0 & 100.0 \\ 
  1000 &  0 &  0 & 0.2 & 0.2 & 5.6 & 4.9 & 89.8 & 77.0 & 100.0 & 99.6 & 100.0 & 100.0 \\ 
   \hline
\end{tabular}
\caption{\it Simulated rejection rates (in \%) of the Jackknife test  \eqref{test2} and pivotal test  \eqref{test2a}   for the  relevant hypotheses  \eqref{det112} in  Model 1.} \label{tb:rel1}
\end{table}
In the ``interior of the null hypothesis'', where  $ \mathcal M^2> \Delta $, the rejection rates of both tests quickly decrease to $0$  as $\Delta$ decreases and the sample size increases. Note that this is a desirable behavior for a test of the composite hypotheses \eqref{det112} as one has to assure for a type I error that 
$$
\sup_{{\cal M}^2 \geq \Delta }
\mathbb{P} (\text{ ``rejection''}  ) \leq \alpha .
$$
Thus, the level of the test is calibrated at the ``boundary'' of null hypothesis ${\cal M}^2 = \Delta $  to be close to $\alpha$.     Under the alternative, where   $  \mathcal M^2 < \Delta $,  the rejection rates of both tests increase to $1$ as $\Delta$ increases and the sample size increases. The rejection rates of the Jackknife test  \eqref{test2} are usually higher than those of the pivotal test  \eqref{test2a}.

\begin{table}[ht]
\centering
\begin{tabular}{ccccccc|cccccc}
   \hline
 &\multicolumn{6}{c|}{$H_0$}& \multicolumn{6}{c}{$H_1$}\\ 
   \hline
  &\multicolumn{2}{c}{$\Delta = 1.7$}& \multicolumn{2}{c}{$\Delta = 1.8$}& \multicolumn{2}{c|}{$\Delta = \mathcal M^2$}&\multicolumn{2}{c}{$\Delta = 2.3$}& \multicolumn{2}{c}{$\Delta = 2.5$}& \multicolumn{2}{c}{$\Delta = 2.8$} \\
    \hline
$n$ & \eqref{test2}  & \eqref{test2a} & \eqref{test2}  & \eqref{test2a} & \eqref{test2}  & \eqref{test2a} & \eqref{test2}  & \eqref{test2a} & \eqref{test2} & \eqref{test2a} & \eqref{test2} & \eqref{test2a}\\ 
  \hline
200 & 0.8 & 0.2 & 3.2 & 1.3 & 8.0 & 4.5 & 36.1 & 24.2 & 60.8 & 43.0 & 85.3 & 69.1 \\ 
  300 & 0.2 & 0.3 & 1.7 & 1.1 & 8.5 & 4.4 & 46.8 & 32.1 & 72.9 & 55.5 & 93.6 & 81.0 \\ 
  400 & 0.0 & 0.0 & 1.7 & 1.4 & 6.0 & 3.8 & 57.3 & 41.7 & 82.8 & 67.5 & 97.4 & 90.9 \\ 
  500 & 0.2 & 0.1 & 1.0 & 0.7 & 6.9 & 4.6 & 63.3 & 46.3 & 90.6 & 76.7 & 99.5 & 94.4 \\ 
  600 & 0.1 & 0.1 & 0.8 & 1.1 & 5.9 & 4.8 & 66.7 & 51.9 & 93.2 & 81.4 & 99.7 & 96.4 \\ 
  800 & 0.0 & 0.0 & 0.1 & 0.3 & 5.3 & 3.8 & 78.4 & 59.9 & 97.7 & 89.0 & 100.0 & 98.7 \\ 
  1000 & 0.0 & 0.0 & 0.2 & 0.2 & 4.9 & 4.8 & 82.8 & 71.0 & 98.7 & 92.6 & 100.0 & 99.3 \\ 
   \hline
\end{tabular}
\caption{\it Simulated rejection rates (in \%) of Jackknife test \eqref{test2}  and the pivotal test  \eqref{test2a}   for the  relevant hypotheses \eqref{det112} in  Model 2.}\label{tb:rel2}
\end{table}

In summary, for testing for relevant deviations from sphericity from independent identically distributed data, the Jackknife test exhibits some advantages compared to the pivotal test if the sample size is sufficiently large. However, we emphasize that this observation can only be made for independent data.  In the presence of dependencies, the Jackknife estimator does not yield a valid testing procedure. On the other hand, a careful  inspection of the proofs in the appendix shows that the asymptotic statements in Section \ref{sec4} remain valid for stationary processes under appropriate mixing  \citep{Bradley.2007}, physical dependence \citep{Wu2005} or $m$-approximability \citep{HrmannKok} conditions.  Consequently, Corollary \ref{cor41} remains valid as well, and the pivotal test \eqref{test2a} has asymptotic level $\alpha$ and is consistent for the hypotheses \eqref{det112}.  

We illustrate this fact by a small simulation for the tests \eqref{test2} and \eqref{test2a} in a model with dependent data. To be precise,  let $Z_i=(z_{i,1}, z_{i,2}, z_{i,3})^{\top}$, where 
\begin{align}
     z_{i,j} = 0.3 z_{i-i,j} + \epsilon_{i,j}, \quad j=1,2,3,  
\end{align}
and the  $\epsilon_{i,j}'s $ are independent standard normal distributed random variables. The data $( Y_i)_{i=1}^n$ is then generated by 
\begin{align}
     Y_i = \Sigma^{1/2} Z_i \times \sqrt{1-0.3^2 } + \mu,
\end{align}
where $ \Sigma$ and $\mu$ are defined in Model 1.  We display in Table \ref{tabdep} the rejection probabilities of both tests for dependent data. In the interior of the null hypothesis  and under the alternative we  observe a quantitative similar behavior as for independent data. However, at the boundary ($\mathcal{M}^2= \Delta$) the Jackknife test \eqref{test2} does not keep its nominal level $5\%$ (in all considered cases). On the other hand, the pivotal test \eqref{test2a} yields a very good approximation of the nominal level.

\begin{table}[ht]
\centering
\begin{tabular}{ccccccc|cccccc}
  \hline
 &\multicolumn{6}{c|}{$H_0$}& \multicolumn{6}{c}{$H_1$}\\ 
   \hline
  &\multicolumn{2}{c}{$\Delta = 0.7$}& \multicolumn{2}{c}{$\Delta = 0.9$}& \multicolumn{2}{c|}{$\Delta = \mathcal M^2$}&\multicolumn{2}{c}{$\Delta = 1.1$}& \multicolumn{2}{c}{$\Delta = 1.3$}& \multicolumn{2}{c}{$\Delta = 1.5$} \\
    \hline
$n$ & \eqref{test2}  & \eqref{test2a} & \eqref{test2}  & \eqref{test2a} & \eqref{test2}  & \eqref{test2a} & \eqref{test2}  & \eqref{test2a} & \eqref{test2} & \eqref{test2a} & \eqref{test2} & \eqref{test2a}\\ 
  \hline
200 &  0 &  0 & 2.2 & 0.8 & 7.0 & 3.5 & 38.1 & 23.4 & 86.4 & 63.1 & 98.6 & 89.7 \\ 
  300 &  0 &  0 & 2.4 & 1.1 & 9.0 & 4.7 & 48.2 & 30.7 & 93.3 & 81.2 & 99.9 & 97.1 \\ 
  400 &  0 &  0 & 1.3 & 0.9 & 8.2 & 4.3 & 55.9 & 39.8 & 97.3 & 87.5 & 100.0 & 98.6 \\ 
  500 &  0 &  0 & 1.1 & 1.2 & 7.6 & 4.8 & 64.7 & 44.8 & 99.6 & 92.8 & 100.0 & 99.6 \\ 
  600 &  0 &  0 & 1.4 & 0.4 & 10.1 & 5.4 & 70.1 & 52.1 & 99.8 & 94.7 & 100.0 & 99.7 \\ 
  800 &  0 &  0 & 0.9 & 0.5 & 7.8 & 5.0 & 79.9 & 59.2 & 100.0 & 98.1 & 100.0 & 100.0 \\ 
  1000 &  0 &  0 & 0.9 & 0.4 & 8.4 & 5.4 & 86.0 & 66.8 & 100.0 & 99.5 & 100.0 & 100.0 \\ 
   \hline
\end{tabular}
\caption{\it Simulated rejection rates (in \%) of the Jackknife test \eqref{test2}  and the pivotal test \eqref{test2a}  for the  relevant hypotheses \eqref{det112} in the case of dependent data.\label{tabdep}}
\end{table}

\subsection{Comparison with other methods}\label{sec:compare}
In this section, we compare the minimum distance approach with alternative  methods for the  detection of deviations from sphericity. As all other papers focus on testing for exact ellipticity or sphericity, we consider the hypotheses \eqref{det208} and the test \eqref{testexact} in Remark \ref{rem31}. 
Note that under the null hypothesis ${\cal M}^2=0$, the estimator $\hat{\mathcal M}_n^2$ is unbiased (see Proposition \ref{prop:meanvar}(i)). Therefore, no bias correction as discussed in  Remark \ref{rm:bias} is required in this situation. For the selection of bandwidths, we balance the squared bias under the alternative hypothesis, i.e., $\kappa^{-2} + h^4$ and the variance under Theorem \ref{thm:normal}, i.e., $\kappa^{(p-1)/2}/[n(n-1)h]$, and for $p=3$  use 
\begin{align}
          h = 1.5 (n(n-1))^{-1/7} ~, ~~
      \kappa = 1.5 (n(n-1)h)^{1/3}.
\end{align} 
In the following, we compare our test with other tests for the exact hypothesis \eqref{det208}   proposed by   \cite{SCHOTT2002}, \cite{MPQ2002}, \cite{HP2007},  \cite{PG2008}, \cite{babic2021} and \cite{tang2024nonparametric}. 
  We consider the null hypothesis, where $Y$ follows a $3$-dimensional standard normal  distribution. For the alternative we changed the first component of $Y$ to  a $(\chi_2^2 - 2)/2$-distribution. 

In Table \ref{tb:exacth0}, we investigate the simulated rejection rate under the null hypothesis. The new test \eqref{testexact} provides a good approximation of the nominal level $\alpha= 10\%$, $5\%$. The alternative tests exhibit a similar behavior under the null hypothesis, except for the test of \cite{tang2024nonparametric}, which turns out to be conservative. The results under the alternative  are displayed in Table \ref{tb:exacth1}.  Although the new  test \eqref{testexact} does not require Gaussianity or the existence of $4$-th moments, the empirical rejection rates under the alternative are larger  than those of the tests \cite{HP2007}, \cite{MPQ2002}, \cite{SCHOTT2002},  which require such assumptions,  and  comparable to the optimal parametric and semiparametric tests  proposed by  \cite{PG2008} and \cite{babic2021}.
Moreover,  the new test \eqref{testexact} for the hypotheses \eqref{det208}  is more powerful in this scenario than the test based on kernel embedding of probabilities proposed by  
\cite{tang2024nonparametric}.


\begin{table}[ht]
\centering
\begin{tabular}{l|rrr|rrr|rrr|rrr}
 \hline
 \hline
 &\multicolumn{3}{c|}{\eqref{det208}}& \multicolumn{3}{c|}
{TL} &\multicolumn{3}{c|} {HP} & \multicolumn{3}{c} {MPQ} \\
\hline
\hline
$n$$|\alpha$ & 10\% & 5\% & 1\% & 10\% & 5\% & 1\% & 10\% & 5\% & 1\%  & 10\% & 5\% & 1\%  \\ 
 \hline
200	&8.6 & 5.0 & 1.0 & 3.7 & 1.7 & 0.3 & 10.2& 4.4 & 0.5 & 10.9 & 5.5 & 0.5 \\ 
  300&	10.1 & 5.5 & 1.2 & 3.9 & 1.9 & 0.2 & 10.1 & 5.2 & 1.3 & 10.6 & 5.4 & 1.2 \\ 
  400&	9.2 & 4.6 & 1.4 & 4.1 & 1.5 & 0.3 & 8.0 & 4.0 & 0.5 & 10.0 & 5.9 & 1.1 \\ 
   \hline
    \hline
 &\multicolumn{3}{c|}{\eqref{det208}}& \multicolumn{3}{c|}
{PG} &\multicolumn{3}{c|} {Schott} & \multicolumn{3}{c} {BGH} \\
\hline
\hline
$n$$|\alpha$ & 10\% & 5\% & 1\% & 10\% & 5\% & 1\% & 10\% & 5\% & 1\%  & 10\% & 5\% & 1\%  \\ 
 \hline
200	&8.6 & 5.0 & 1.0 & 11.0 & 5.1 & 1.2 & 10.1 & 4.8 & 0.6 & 10.9 & 4.7 & 0.9 \\ 
  300	&10.1 & 5.5 & 1.2 & 10.2 & 5.8 & 1.5 & 9.9 & 4.3 & 0.8 & 11.3 & 5.2 & 1.2 \\ 
  400&	9.2 & 4.6 & 1.4 & 12.2 & 7.2 & 1.3 & 10.8 & 5.4 & 1.2 & 10.9 & 5.9 & 0.9 \\ 
   \hline
   \hline
\end{tabular}
\caption{\it Simulated rejection rates (in \%) of the test {\eqref{det208}} and  the  tests of \cite{tang2024nonparametric} (TL), 
\cite{HP2007} (HP), \cite{MPQ2002} (MPQ), \cite{PG2008} (PG), \cite{SCHOTT2002} (Schott) and \cite{babic2021} (BGH)
for the null hypothesis of exact sphericity.  
The sample sizes are $200$, $300$ and $400$. The  nominal levels are   $10$\%, $5$\% and $1$\%. }
\label{tb:exacth0}
\end{table}

\begin{table}[!ht]
\centering
\begin{tabular}{l|rrr|rrr|rrr|rrr}
 \hline
 \hline
 &\multicolumn{3}{c|}{\eqref{det208}}& \multicolumn{3}{c|}
{TL} &\multicolumn{3}{c|} {HP} & \multicolumn{3}{c} {MPQ} \\
\hline
\hline
$n$$|\alpha$ & 10\% & 5\% & 1\% & 10\% & 5\% & 1\% & 10\% & 5\% & 1\%  & 10\% & 5\% & 1\%  \\ 
 \hline
200	&99.5 & 98.1 & 91.3 & 98.5 & 92.7 & 45.0 & 94.6 & 89.9 & 71.8 & 93.6 & 88.4 & 73.0 \\ 
  300&	99.9 & 99.9 & 99.8 & 100.0 & 99.9 & 93.8 & 99.9 & 99.0 & 94.8 & 99.4 & 98.9 & 94.6 \\ 
  400&	100.0& 100.0 & 100.0 & 99.9 & 99.9 & 99.7 & 100.0 & 99.9 & 99.7 & 99.9 & 99.9 & 99.4 \\ 
   \hline
    \hline
 &\multicolumn{3}{c|}{\eqref{det208}}& \multicolumn{3}{c|}
{PG} &\multicolumn{3}{c|} {Schott} & \multicolumn{3}{c} {BGH} \\
\hline
\hline
$n$$|\alpha$ & 10\% & 5\% & 1\% & 10\% & 5\% & 1\% & 10\% & 5\% & 1\%  & 10\% & 5\% & 1\%  \\ 
 \hline
200	&99.5 & 98.1 & 91.3 & 100.0& 100.0& 100.0& 75.2 & 64.2 & 42.3 & 100.0& 100.0& 99.9 \\ 
  300&	99.9 & 99.90 & 99.80 & 100.0& 100.0& 100.0& 91.4 & 82.9 & 63.2 & 100.0& 100.0& 100.0 \\ 
  400&	100.0& 100.0 & 100.0 & 100.0& 100.0& 100.0& 95.7 & 91.6 & 79.0 & 100.0& 100.0& 100.0 \\ 
   \hline
   \hline
\end{tabular}
\caption{\it Simulated rejection rates (in \%) of the test {\eqref{det208}} and  the  tests of \cite{tang2024nonparametric} (TL), 
\cite{HP2007} (HP), \cite{MPQ2002} (MPQ), \cite{PG2008} (PG), \cite{SCHOTT2002} (Schott) and \cite{babic2021} (BGH)
for the hypotheses of exact sphericity under the alternative hypothesis.  
The sample sizes are $200$, $300$ and $400$. The  nominal levels are   $10$\%, $5$\% and $1$\%. }
\label{tb:exacth1}
\end{table}

\subsection{Real data example}

To illustrate the potential applications of our approach, we analyze  the  log-returns  of the daily exchange
rate (close price) of the Yen to the Dollar and the Pound to the Euro from January 2, 2009, to December 31, 2009, which has also been investigated in \cite{einmahl2012} using a test for independent data. We downloaded the data from Yahoo Finance via the R package ``quantmod" and obtained 260 entries for the log-returns, as the transactions of the currency can be conducted on weekdays, including holidays. The mean of the bivariate series turns out to be $-9.4 \times 10^{-5}$ and $2.7 \times 10^{-4}$. In Figure \ref{fig:exchange}, we plot the absolute autocorrelations of different lags of both of the series.  We  found that there exists some evidence of serial dependence, for example,  the lag 12 of the Yen to Dollar and the lag 6 of the Pound to Euro. 
\begin{figure}[h]
    \centering
    \includegraphics[width=0.7\linewidth]{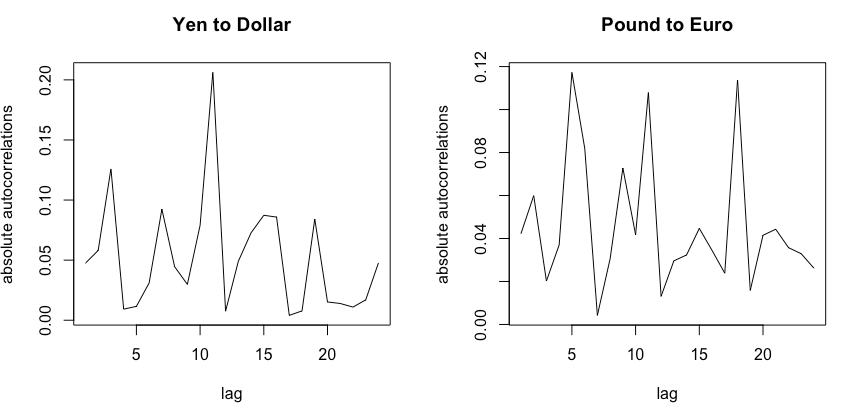}
    \caption{\it Absolute autocorrelations of the log-returns of the daily exchange
rate of the Yen to the Dollar and the Pound to the Euro from January 2nd, 2009, to December 31st, 2009.}
    \label{fig:exchange}
\end{figure}
The test developed by  \cite{einmahl2012} refers to the classical hypothesis of exact sphericity  (see equation \eqref{det208}) and does   not reject the null hypothesis. However, this does  not imply the null hypothesis is true, since only the type I error of deciding against exact sphericity (although it holds) is controlled. In order to control the error of deciding for a spherically symmetric distribution, we consider the hypotheses \eqref{det111}  and 
apply the  pivotal test, which additionally defines a valid inference procedure  for  dependent data  (see the discussion at the end of Section \ref{sec52}. We use the same tuning parameters as in Model 1 and obtain $\mathcal{
\hat M}_n^2 = 1.443$ as the estimator of  ${\cal M}^2$. According to Remark \ref{rem31} (a), we calculate the data dependent threshold $\hat \Delta_{0.05}$, which is $ 1.448$. The result shows that at the significance level of $0.05$, when the chosen $\Delta$ is smaller than  $\hat \Delta_{0.05}$, a model should be considered beyond spherical symmetry, and when the chosen $\Delta$ is larger than $\hat  \Delta_{0.05}$, one favors the simpler model with spherical symmetry. These results indicate that the assumption of sphericity is hard to justify for this data set.


 \bigskip

 \noindent
{\bf  Acknowledgements} 
 This work was partially supported by 
 DFG Research unit 5381 \textit{Mathematical Statistics in the Information Age} (Project number 460867398) and by 
 TRR 391 \textit{Spatio-temporal Statistics for the Transition of Energy and Transport} (Project number 520388526) funded by the Deutsche Forschungsgemeinschaft (DFG, German Research Foundation).

 \newpage

\appendix

\section{Preliminary technical results}

In this section, we present and prove several preliminary results, which will be essential for the proofs of the statements in Section \ref{sec3} and \ref{sec4}.
Our first result provides some properties of  the Fisher-von-Mises distribution.


\begin{lemma}
\label{lm:rate}
    For the Fisher-von-Mises distribution in \eqref{eq:langevin} we have for any $j \geq 1$, $p \geq 2$
\begin{align}
    \label{det200}
  b_j(\kappa) &= \omega_{p-2} \int_0^{\pi} L(\kappa \cos \theta) (\sin \theta)^{p-2} \theta^j \mathrm d \theta  \sim   a_j(p)\kappa^{-j/2}   \end{align}
 as $\kappa \to \infty $, where  $a_j(p) = 2^{j/2}\Gamma((p+j-1)/2)/\Gamma((p-1)/2) $. Similarly, 
\begin{align}
    \label{det201}
c_j(\kappa) &= \omega_{p-2} \int_0^{\pi} L^j(\kappa \cos \theta) (\sin \theta)^{p-2} \mathrm d \theta \sim    d_j(p) \kappa^{(j-1)(p-1)/2} 
\end{align} 
as $\kappa \to \infty $, where $d_j(p) = 2^{(1-j) (p-1)/2} j^{-(p-1)/2} 
 \pi^{(1-j) (p-1)/2}$. Moreover, for any $t>0$ we have 
 \begin{align}
\label{det203}
{L^j(\kappa t) \over L(j\kappa t)} \sim  \kappa^{(j-1)(p-1)/2} j^{-p/2+1/2} (2\pi)^{(j-1)(p-1)/2}
 \end{align}
  as $\kappa \to \infty $.
\end{lemma}

\begin{proof}
    By equation (9)  in \cite{DiMarzio03042014}, we have 
    \begin{align}
       b_j (\kappa) \sim  \frac{2^{j/2}\Gamma((p+j-1)/2)}{\kappa^{j/2} \Gamma((p-1)/2) }
    \end{align}
    as $\kappa \to \infty$ and  with  the notation 
     $a_j(p) = 2^{j/2}\Gamma((p+j-1)/2)/\Gamma((p-1)/2) $
    the result  \eqref{det200} follows.
    For a proof of \eqref{det201} we note that the modified Bessel function  in \eqref{eq:bessel} satisfies  $\mathcal I_v (\kappa) \sim  e^{\kappa}/\sqrt{2\pi \kappa}$, as $\kappa \to \infty$ (see equation (9.6.18) on Page 376 of  \cite{abramowitz1968handbook}). Therefore, recalling the definition of the Fisher-von-Mises distribution in \eqref{eq:langevin},  we obtain
\begin{align}
    \label{det204}
L(\kappa t) \sim  \kappa^{p/2-1} \{(2\pi)^{p/2}  e^{\kappa} /\sqrt{2 \pi \kappa}\}^{-1} e^{ \kappa t},
\end{align}
which yields (with the substitution $t = \cos \theta$)  
     \begin{align}
        c_j (\kappa) &= \omega_{p-2} \int_{-1}^1 L^j(\kappa t) (1 - t^2)^{(p-3)/2}\dd t\\ 
        &  \sim  \omega_{p-2} \kappa^{jp/2-j/2} \{(2\pi)^{(p-1)/2} e^k\}^{-j} \int_{-1}^1 e^{j\kappa t}(1 - t^2)^{(p-3)/2}\dd t.
    \end{align}
  Observing that  
     \begin{align}
      \int_{-1}^1 e^{j\kappa t} (1 - t^2)^{(p-3)/2}\dd t  
       ~ & = \int_0^2 e^{j\kappa(1-t)} \{t(2-t)\}^{(p-3)/2} \dd t\\ 
        &= \int_{0}^2 e^{j\kappa}e^{-j\kappa  t} t^{(p-3)/2} (2-t)^{(p-3)/2} \dd t\\
        &\sim  2^{(p-3)/2} e^{j\kappa} (j\kappa)^{-(p-1)/2}   \int_{0}^{2j\kappa}  e^{- 
     t} t^{(p-3)/2}\dd t\\
     &\sim  2^{(p-3)/2} e^{j\kappa} (j\kappa)^{-(p-1)/2}   \Gamma\{(p-1)/2\}, \label{eq:ejkint}
    \end{align}
we have for $j 
    \geq 1$  
 \begin{align}
      c_j (\kappa) & \sim    \omega_{p-2} (\kappa/(2\pi))^{pj/2-j/2}
 2^{(p-3)/2}  (j\kappa)^{-(p-1)/2}   \Gamma \{(p-1)/2\} \\ 
 &\sim    
 \kappa^{(j-1)(p-1)/2} 2^{(1-j) (p-1)/2} j^{-(p-1)/2} 
 \pi^{(1-j) (p-1)/2},
 \end{align}
  where we used $
  \omega_{p-2} = 2\pi^{(p-1)/2}/\Gamma((p-1)/2)$ in the last step. The assertion \eqref{det201} now follows with the notation 
  $d_j(p) = 2^{(1-j) (p-1)/2} j^{-(p-1)/2} 
 \pi^{(1-j) (p-1)/2}$. 
   Finally, \eqref{det203} is obtained by  a direct calculation  using the expansion  \eqref{det204}.
\end{proof}

\noindent
We will use Lemma \ref{lm:rate} for the calculation of the moments of the statistic \eqref{eq:Ustat}, which is complicated due to the definition of the kernel $H_n$. We present two auxiliary results, which are required as intermediate steps in these calculations.

\begin{lemma} 
If  Assumption \ref{ass:kernel} and \ref{ass:density} are satisfied, we have
    \begin{align}
        h^{-1}c_1^{-1} (\kappa) \E \Big [ K \Big (\frac{U_i-U_j}{h}\Big )  L(\kappa V_i^{\top} V_j) \Big ] &= \int_{\mathbb{S}^{p-1}} \int_{\R^+} f^2(u,v) \mathrm du \ \omega_{p-1}(\mathrm  dv)\\ & 
        +  \frac{h^2\phi_2(K) }{2}  \int_{\mathbb{S}^{p-1}} \int_{\R^+}  \left. f(u_2, v_1) \frac{\partial^2 f(u, v_1) }{\partial u^2} \right|_{u = u_2} \dd u_2 \omega_{p-1}(\dd v_1)\\ &+ \frac{\omega_{p-2}b_2(\kappa)}{2c_1(\kappa)(p-1)} \int_{\mathbb{S}^{p-1}} \int_{\R^+} f(u_2, v_1)   \mathrm{tr}\{ D^2_{f}(u_2, v_1)\}      \dd u_2 \omega_{p-1}(\dd v_1)\\ &+ O(c_1^{-1} (\kappa) b_3(\kappa) + b_2(\kappa) h^2 + h^3),\\
        h^{-1}c_1^{-2}(\kappa) \E  \Big [ K^2\Big (\frac{U_i-U_j}{h}\Big )  L^2(\kappa V_i^{\top} V_j) \Big ] &= c_1^{-2}(\kappa) c_2(\kappa) \psi_2(K)  \int_{\mathbb{S}^{p-1}} \int_{\R^+} f^2(u,v) \mathrm du \mathrm  dv \\ &+ O(c_1^{-2}(\kappa)c_2(\kappa) h^2 + c_1^{-2}(\kappa) \kappa^{(p-1)/2}b_2(2\kappa)),\\
         h^{-1}c_1^{-1}(\kappa) \E \Big [ K^2\Big (\frac{U_i-U_j}{h}\Big ) L(\kappa V_i^{\top} V_j) \Big ] &=  \psi_2(K)  \int_{\mathbb{S}^{p-1}} \int_{\R^+} f^2(u,v) \mathrm du \mathrm  dv + O(h^2 + c_1^{-1}(\kappa) b_2(\kappa)),\\ h^{-1} \E \Big [ 
         K \Big  (\frac{U_i-U_j}{h}\Big ) \Big ] &= \int_{\R^+} f_U^2(u) du  + \frac{h^2 \phi_2(K) }{2} \int_{\R^+} f_U (u)f^{\prime \prime}_U (u) du  + O(h^3), \\
         h^{-1} \E  \Big [ K^2 \Big (\frac{U_i-U_j}{h}\Big ) \Big ]&= \psi_2(K) \int_{\mathbb R^+} f_U^2(u) du  +O(h^2),
    \end{align}
    where $\psi_j(K) = \int_{-1}^1 K^j(u) \dd u$, $\phi_j(K) = \int_{-1}^1 x^j K(u) \dd u$.
    \label{lm:kernelse}
\end{lemma}
\begin{proof}
    \begin{align}
       & \E \Big [ K \Big (\frac{U_i-U_j}{h} \Big ) L(\kappa V_i^{\top} V_j) \Big ] \\ &= \int_{\mathbb{S}^{p-1}} \int_{\mathbb{S}^{p-1}} \int_{\mathbb R^+} \int_{\mathbb R^+} K\left(\frac{u_1 - u_2}{h}\right) L(\kappa v_1^{\T} v_2) f(u_1, v_1) f(u_2, v_2) \dd u_1 \dd u_2\  \omega_{p-1}(\dd v_1) \ \omega_{p-1}(\dd v_2). \label{eq:EKL}  
    \end{align}
    By  a change of variables, that is  $u = (u_1 - u_2)/h$, $\cos 
    \theta = v_1^{\T} v_2$,  we obtain 
    \begin{align}
        u_1 = hu + u_2, \quad v_2 = v_1\cos \theta+ \xi \sin \theta, \quad \xi \perp v_1. \label{eq:changev}
    \end{align}
  Note that  $\xi \in \Omega_{v_1}:=\{  \xi \in \mathbb{S}^{p-1} :~ \xi \perp v_2 \}$, so that 
    \begin{align}
        \omega_{p-1}(dv_2) = (\sin \theta)^{p-2} \dd \theta \ \omega_{p-2}(\dd \xi).\label{eq:changem}
    \end{align}
    Then, with \eqref{eq:changev} and \eqref{eq:changem}, we can rewrite \eqref{eq:EKL} into 
    \begin{align}
         & \E \Big [ K \Big (\frac{U_i-U_j}{h} \Big ) L(\kappa V_i^{\top} V_j) \Big ]   = h \int_{\Omega_{v_1}} \int_0^{\pi} \int_{\mathbb{S}^{p-1}} \int_{\mathbb R^+} \int_{\mathbb R^+} K\left(u\right) L(\kappa \cos \theta) f(hu + u_2, v_1)\\&\times  f(u_2, v_1\cos \theta+ \xi \sin \theta) \dd u \dd u_2 \omega_{p-1}(\dd v_1) \ (\sin \theta)^{p-2} \dd \theta \   \omega_{p-2}(\dd \xi). \label{eq:EKL1}  
    \end{align}
    Recall the  notations
    \begin{align}
   & b_j(\kappa) = \omega_{p-2} \int_0^{\pi} L(\kappa \cos \theta) (\sin \theta)^{p-2} \theta^j \mathrm d \theta,\\ 
   & c_j(\kappa) = \omega_{p-2} \int_0^{\pi} L^j(\kappa \cos \theta) (\sin \theta)^{p-2} \mathrm d \theta 
\end{align}
and  note  that (see the proof  of Theorem 3.1 in  \cite{DiMarzio03042014}) 
\begin{align}
    \int_{\Omega_{v_1}} \xi \omega_{p-2}(\dd \xi) = \mf 0_p,  \quad \int_{\Omega_{v_1}} \xi \xi^{\T} \omega_{p-2}(\dd \xi) = { \omega_{p-2}  \over (p-1)}  (\mf I_p -  v_1 v_1^{\T}), \label{eq:xi1}
\end{align}
where  $\mf 0_p$ and $\mf I_p$ denote the $p$-dimensional origin  and $p \times p$ identity matrix, respectively. Moreover,
$f(u,v_1)$ is defined on 
$\R^+ \times \mathbb{S}^{p-1}$
and 
\begin{align}
    v_1 ^{\T} D^2_{f}(u,v_1) v_1  = 0 \label{eq:dfv}
\end{align}
\citep[see the proof of  Theorem 3.1  in][]{DiMarzio03042014},
which yields  
\begin{align}
    \int_{\Omega_{v_1}} \xi^{\T
    } D_{f} (u_2, v_1) \omega_{p-2}(\dd \xi) &= 0, \label{eq:Dxi1}\\     \int_{\Omega_{v_1}} \xi^{\T
    } D^2_{f} (u_2, v_1) \xi \ \omega_{p-2}(\dd \xi) & =\int_{\Omega_{v_1}} \mathrm{tr} \{D^2_{f} (u_2, v_1) \xi \xi^{\T
    }\}  \ \omega_{p-2}(\dd \xi) \\ 
    &= \mathrm{tr} \{D^2_{f} (u_2, v_1) (\omega_{p-2} (p-1)^{-1} (\mf I_p -  v_1 v_1^{\T}))\}  \\
    & = \omega_{p-2} (p-1)^{-1} \mathrm{tr} \{D^2_{f} (u_2, v_1) \}. \label{eq:Dxi2}
\end{align}
    Using the Taylor expansions in \eqref{eq:tayloru} and \eqref{eq:taylorv}, we therefore obtain for  \eqref{eq:EKL1}
    \begin{align}
       &h^{-1} \E \Big [ K \Big (\frac{U_i-U_j}{h} \Big ) L(\kappa V_i^{\top} V_j) \Big ]  \\
        & = \int_{\mathbb{S}^{p-1}} \int_{\mathbb R^+} \left\{\int_{-1}^1 K\left(u\right) [f(u_2, v_1) + h u \left.\frac{\partial f(u, v_1) }{\partial u} \right|_{u = u_2} +  \frac{h^2 u^2}{2} \left.\frac{\partial^2 f(u, v_1) }{\partial u^2} \right|_{u = u_2}  + O(h^3)]du   \right\} \\ 
        & \times \left\{\int_{\Omega_{v_1}} \int_0^{\pi}  L\left(\kappa \cos \theta\right) [f(u_2, v_1)+ \theta \xi^{\T} D_{f}(u_2, v_1) 
        \right . \\
        & ~~~~~~~~ \left . + \frac{\theta^2}{2}\xi^{\T} D^2_{f}(u_2, v_1) \xi + O(\theta^3)]\ (\sin \theta)^{p-2} \dd \theta \   \omega_{p-2}(\dd \xi)  \right\}\dd u_2 \omega_{p-1}(\dd v_1)\\
        & = \int_{\mathbb{S}^{p-1}} \int_{\mathbb R^+} \left\{f(u_2, v_1) + \frac{h^2\phi_2(K) }{2}  \left.\frac{\partial^2 f(u, v_1) }{\partial u^2} \right|_{u = u_2}  + O(h^3)  \right\}\\ 
        & \times \left\{ c_1(\kappa) f(u_2, v_1) + 
        \frac{\omega_{p-2}}{2(p-1)} \mathrm{tr}\{ D^2_{f}(u_2, v_1)\} \int_0^{\pi} L\left(\kappa \cos \theta\right) \theta^2 (\sin \theta)^{p-2} \dd \theta \right.\\ & ~~~~~~~~
        \left. + \int_{\Omega_{v_1}}  \int_0^{\pi}  L\left(\kappa \cos \theta\right) O( \theta^3 ) (\sin \theta)^{p-2} \dd \theta \omega_{p-2}(\dd \xi)  \right\}\dd u_2 \omega_{p-1}(\dd v_1)\\ 
        & =  c_1(\kappa) \int_{\mathbb{S}^{p-1}} \int_{\mathbb R^+} f^2(u_2, v_1)\dd u_2 \omega_{p-1}(\dd v_1)\\&
        ~~~~~~~~ 
         +  c_1(\kappa) \frac{h^2\phi_2(K) }{2}  \int_{\mathbb{S}^{p-1}} \int_{\mathbb R^+}  \left. f(u_2, v_1) \frac{\partial^2 f(u, v_1) }{\partial u^2} \right|_{u = u_2} \dd u_2 \omega_{p-1}(\dd v_1) \\
         & ~~~~~~~~
        +  \frac{\omega_{p-2}b_2(\kappa)}{2(p-1)} \int_{\mathbb{S}^{p-1}} \int_{\mathbb R^+} f(u_2, v_1)   \mathrm{tr}\{ D^2_{f}(u_2, v_1)\}      \dd u_2 \omega_{p-1}(\dd v_1)+ O(b_3(\kappa) + h^3 + b_2(\kappa)h^2), 
    \end{align}
    where  we have used $\int_{-1}^{1} uK(u)\dd u = 0$, \eqref{eq:Dxi1} and \eqref{eq:Dxi2} in the second equality. \\    
    Recalling  the notation 
    \begin{align}
        c_2(\kappa) = \omega_{p-2} \int_0^{\pi} L^2(\kappa \cos \theta) (\sin \theta)^{p-2} \mathrm d \theta,
    \end{align}
    it follows by similar arguments and \cref{lm:rate} that 
    \begin{align}
       & h^{-1} \E \Big [ K^2\left(\frac{U_i-U_j}{h}\right) L^2(\kappa V_i^{\top} V_j)  \Big ] \\
         & = \int_{\mathbb{S}^{p-1}} \int_{\mathbb R^+} \left\{\int_{-1}^1 K^2\left(u\right) [f(u_2, v_1) + h u \left.\frac{\partial f(u, v) }{\partial u} \right|_{u = u_2} + O(h^2)]du   \right\} \\ 
        & \times \left\{\int_{\Omega_{v_1}} \int_0^{\pi}  L^2\left(\kappa \cos \theta\right) [f(u_2, v_1)+ \theta \xi^{\T} D_{f}(u_2, v_1) \right. \\
        &  ~~~~~~~ \left . + \frac{\theta^2}{2}\xi^{\T} D^2_{f}(u_2, v_1) \xi + O(\theta^3)]\ (\sin \theta)^{p-2} \dd \theta \   \omega_{p-2}(\dd \xi)  \right\}\dd u_2 \omega_{p-1}(\dd v_1)\\
        & = \int_{\mathbb{S}^{p-1}} \int_{\mathbb R^+} \left\{\psi_2(K)  f(u_2, v_1) + O(h^2)  \right\} \times \left\{ c_2(\kappa) f(u_2, v_1)  
        \right. \\
        &  ~~~~~~~ \left . +
    \frac{\omega_{p-2}}{2(p-1)} \mathrm{tr}\{ D^2_{f}(u_2, v_1)\} \int_0^{\pi} L^2\left(\kappa \cos \theta\right) \theta^2 (\sin \theta)^{p-2} \dd \theta  \right.\\ &  ~~~~~~~  \left.+ \int_{\Omega_{v_1}}  \int_0^{\pi}  L^2\left(\kappa \cos \theta\right) O( \theta^3 ) (\sin \theta)^{p-2} \dd \theta \omega_{p-2}(\dd \xi)   \right\}\dd u_2 \omega_{p-1}(\dd v_1)\\ 
        & =  c_2(\kappa) \psi_2(K)   \int_{\mathbb{S}^{p-1}} \int_{\mathbb R^+} f^2(u_2, v_1)\dd u_2 \omega_{p-1}(\dd v_1) + O(c_2(\kappa) h^2 + \kappa^{(p-1)/2} b_2(2\kappa)).
    \end{align}
    and 
    \begin{align}
       & h^{-1} \E \Big [  K^2 \Big  (\frac{U_i-U_j}{h}\Big  ) L(\kappa V_i^{\top} V_j) \Big  ]\\
         & = \int_{\mathbb{S}^{p-1}} \int_{\mathbb R^+} \left\{\int_{-1}^1 K^2\left(u\right) [f(u_2, v_1) + h u \left.\frac{\partial f(u, v) }{\partial u} \right|_{u = u_2} + O(h^2)]du   \right\} \\ 
        & \times \left\{\int_{\Omega_{v_1}} \int_0^{\pi}  L\left(\kappa \cos \theta\right) [f(u_2, v_1)+ \theta \xi^{\T} D_{f}(u_2, v_1)  \right. \\
        &  ~~~~~~~ \left . + \frac{\theta^2}{2}\xi^{\T} D^2_{f}(u_2, v_1) \xi + O(\theta^3)]\ (\sin \theta)^{p-2} \dd \theta \   \omega_{p-2}(\dd \xi)  \right\}\dd u_2 \omega_{p-1}(\dd v_1)\\
        & = \int_{\mathbb{S}^{p-1}} \int_{\mathbb R^+} \left\{\psi_2(K)  f(u_2, v_1) + O(h^2)  \right\} \times \left\{  c_1(\kappa) f(u_2, v_1)  \right.\\ &  ~~~~~~~  \left. 
        +\frac{\omega_{p-2}}{2(p-1)} \mathrm{tr}\{ D^2_{f}(u_2, v_1)\} \int_0^{\pi} L\left(\kappa \cos \theta\right) \theta^2 (\sin \theta)^{p-2} \dd \theta \right.\\ & ~~~~~~~ \left.+ \int_{\Omega_{v_1}}  \int_0^{\pi}  L\left(\kappa \cos \theta\right) O( \theta^3 ) (\sin \theta)^{p-2} \dd \theta \omega_{p-2}(\dd \xi)   \right\}\dd u_2 \omega_{p-1}(\dd v_1)\\ 
        & =  c_1 (\kappa) \psi_2(K)   \int_{\mathbb{S}^{p-1}} \int_{\mathbb R^+} f^2(u_2, v_1)\dd u_2 \omega_{p-1}(\dd v_1) + O(c_1(\kappa)h^2 + b_2(\kappa)). 
    \end{align}
    Finally, with simpler arguments, we can also derive 
    \begin{align}
     h^{-1} \E \Big [ K\Big  (\frac{U_i-U_j}{h}\Big ) \Big ] &= \int_{\mathbb R^+} f_U^2(u) du  + \frac{h^2 \phi_2(K) }{2} \int_{\mathbb R^+} f_U (u)f^{\prime \prime}_U (u) du  + O(h^3). \\
         h^{-1} \E \Big [ K^2 \Big  (\frac{U_i-U_j}{h}\Big  ) \Big ] &= \psi_2(K) \int_{\mathbb R^+} f_U^2(u) du  +O(h^2), 
    \end{align}
    which completes the proof of \cref{lm:kernelse}.
\end{proof}

\begin{lemma} 
If  Assumption \ref{ass:kernel} and \ref{ass:density} are satisfied, we have
    \begin{align}
        h^{-1} \E\Big [  K^4 \Big ( \frac{U_i-U_j}{h}\Big  ) L^j(\kappa V_i^{\top} V_j)\ \Big ] &\sim  \psi_4(K)   c_j(\kappa) \int_{\mathbb{S}^{p-1}} \int_{\mathbb R^+} f^2 (u, v) \dd u \ \omega_{p-1}(\mathrm  dv)\\
         h^{-1} \E \Big [ K^4\Big  (\frac{U_i-U_j}{h}\Big ) \Big ] &\sim  \psi_4(K)  \int_{\mathbb R^+} f_U^2 (u) \dd u. 
    \end{align}
    \label{lm:kernel4se}
\end{lemma}
\begin{proof}
Recall the notations 
\begin{align}
    \psi_j (K) =  \int_{-1}^1 K^j\left(u\right) \dd u, \quad c_j(\kappa) = \omega_{p-2} \int_0^{\pi} L^j(\kappa \cos \theta) (\sin \theta)^{p-2} \mathrm d \theta.
\end{align}
 and the Taylor expansions  in \eqref{eq:tayloru} and \eqref{eq:taylorv}, by \eqref{eq:Dxi1} and \eqref{eq:Dxi2}, it follows
     \begin{align}
       & h^{-1} \E \Big [ K^4\Big  (\frac{U_i-U_j}{h}\Big  ) L^j(\kappa V_i^{\top} V_j)\Big ] \\
         & = \int_{\mathbb{S}^{p-1}} \int_{\mathbb R^+} \left\{\int_{-1}^1 K^4\left(u\right) [f(u_2, v_1) + h u \left.\frac{\partial f(u, v) }{\partial u} \right|_{u = u_2} + O(h^2)]du   \right\} \\ 
        & \times \left\{\int_{\Omega_{v_1}} \int_0^{\pi}  L^j\left(\kappa \cos \theta\right) [f(u_2, v_1)+ \theta \xi^{\T} D_{f}(u_2, v_1) + \frac{\theta^2}{2}\xi^{\T} D^2_{f}(u_2, v_1) \xi + O(\theta^3)] \right. \\ 
        & \left. \times (\sin \theta)^{p-2} \dd \theta \   \omega_{p-2}(\dd \xi)  \right\}\dd u_2 \ \omega_{p-1}(\dd v_1)\\
        & = \int_{\mathbb{S}^{p-1}} \int_{\mathbb R^+} \left\{\psi_4(K)  f(u_2, v_1) + O(h^2)  \right\} \times \left\{ c_j(\kappa) f(u_2, v_1) \right.\\ & ~~~~~~~ \left. +
       \frac{\omega_{p-2}}{2(p-1)} \mathrm{tr}\{ D^2_{f}(u_2, v_1)\} \int_0^{\pi} L^j\left(\kappa \cos \theta\right) \theta^2 (\sin \theta)^{p-2} \dd \theta \right.\\ & ~~~~~~~ \left. + \int_{\Omega_{v_1}}  \int_0^{\pi} O(\theta^3) L^j\left(\kappa \cos \theta\right) (\sin \theta)^{p-2} \dd \theta \omega_{p-2}(\dd \xi)  \right\}\dd u_2 \ \omega_{p-1}(\dd v_1)\\ 
        & = \psi_4(K)   c_j(\kappa) \int_{\mathbb{S}^{p-1}} \int_{\mathbb R^+} f^2 (u, v) \dd u \dd v+  o( c_j(\kappa)),
    \end{align}
    where the last line follows from a combination of  \eqref{det200} and \eqref{det203}.
    By similar arguments, we have
    \begin{align}
          h^{-1} \E  \Big [ K^4\Big  (\frac{U_i-U_j}{h} \Big )\Big ]  &=\psi_4(K)  \int_{\mathbb{S}^{p-1}} \int_{\mathbb R^+} f_U^2 (u) \dd u+  o(1),
    \end{align}
  which completes the proof.
\end{proof}

\section{Proof of main results}
\subsection{Proof of \cref{mindist}}\label{proof2.1}
Recalling \eqref{det101} and 
the notation $ h^*(u):=  f_U(u) = \int_{\mathbb{S}^{p-1}} f(u,v) \ \omega_{p-1}(\mathrm dv)$, we obtain for any function $h: \mathbb{S}^{p-1} \to \mathbb{R}$  
    \begin{align}
     &    \int_{\mathcal \mathbb{S}^{p-1}}(f(u,v) - h^*(u) f_0(v)) (h^*(u) f_0(v) - h(u) f_0(v)) \omega_{p-1}(\dd v)\\ 
        & = (h^*(u) - h(u))\int_{\mathcal \mathbb{S}^{p-1}}(f(u,v) - h^*(u) f_0(v))\omega_{p-1}^{-1}  \omega_{p-1}(\dd v) \\
        &  = \omega_{p-1}^{-1}(h^*(u) - h(u)) (h^*(u) - h^*(u) ) = 0.
    \end{align}
    Then, if $h f_0$
    is the density of a spherically symmetric distribution (for some density $h$ on $\mathbb{R}^+$), we obtain  for the squared $L^2$-distance between $f$ and  $h  f_0$
    \begin{align}
         \Psi (h) & :=\int_{\mathbb R^+} \int_{\mathcal \mathbb{S}^{p-1}} ( f(u,v) - h(u) f_0(v))^2 \dd u\ \omega_{p-1}(\dd v) \\ &=  \int_{\mathbb R^+} \int_{\mathcal \mathbb{S}^{p-1}} ( f(u,v)- h^*(u) f_0(v) + h^*(u) f_0(v) - h(u) f_0(v))^2 \dd u\ \omega_{p-1}(\dd v)\\ 
        & = \int_{\mathbb R^+} \int_{\mathcal \mathbb{S}^{p-1}} \{( f(u,v)- h^*(u) f_0(v))^2 + (h^*(u) f_0(v) - h(u) f_0(v))^2 \} \dd u\ \omega_{p-1}(\dd v)\\
         &\geq  \int_{\mathbb R^+} \int_{\mathcal \mathbb{S}^{p-1}} ( f(u,v)- h^*(u) f_0(v))^2 \dd u\ \omega_{p-1}(\dd v).
    \end{align}
    Therefore,  $\Psi (h)$  is minimized for  $h =  h^*$, and 
    \begin{align}
        \mathcal M^2 &= 
        \mathcal M^2 (h^*)  \\
        & = \int_{\mathbb R^+} \int_{\mathcal \mathbb{S}^{p-1}} ( f(u,v)- h^*(u) f_0(v))^2 \dd u\ \omega_{p-1}(\dd v) \\ 
        &= \int_{\mathbb R^+} \int_{\mathcal \mathbb{S}^{p-1}}  f^2(u,v)\dd u\ \omega_{p-1}(\dd v) - 2 \int_{\mathbb R^+} \int_{\mathcal \mathbb{S}^{p-1}}  h^*(u) f(u,v) f_0(v)  \dd u\ \omega_{p-1}(\dd v) \\
        & + \int_{\mathbb R^+} \int_{\mathcal \mathbb{S}^{p-1}} ( f_U(u) f_0(v))^2 \dd u\ \omega_{p-1}(\dd v)\\
        & = \int_{\mathbb R^+} \int_{\mathcal \mathbb{S}^{p-1}}  f^2(u,v)\dd u\ \omega_{p-1}(\dd v) - \omega_{p-1}^{-1}\int_{\mathbb R^+}  f_U^2 (u)  \dd u.
    \end{align}

\subsection{Martingale structure under spherical symmetry}
Recall the representation of the  statistic $\hat{\cal M}_n^2$  in  \eqref{eq:Ustat}, where the kernel $H_n$ is defined in 
\eqref{det206}, and let $\mathcal F_{i-1} = \sigma (Y_1, \ldots, Y_{i-2}, Y_{i-1})$ denote the sigma field generated by the random variables $Y_1, \ldots, Y_{i-2}, Y_{i-1}$.
In the following, we  show that in the case  ${\cal M}^2 =0$   the statistic  $\hat {\cal M}_n $ is a  cumulative sum of martingale differences with respect to the filtration $({\cal F}_i)_{i=1}^n$. This is the cornerstone for the analysis of  the asymptotic properties of the test statistic in the case ${\cal M}^2$=0. 

\begin{lemma} \label{martingale}
 Suppose that Assumption \ref{ass:kernel} and \ref{ass:density} are satisfied and that $\mathcal M^2 =0$, then the sequence
\begin{align}
\label{det207}    
\big (  D_i  \big )_{i=2}^n
= \Big (  \sum_{j=1}^{i-1}H_n(Y_i, Y_j)  \Big )_{i=2}^n 
\end{align} 
is a martingale difference sequence with respect tot the filtration $ (  {\cal F}_i  \big )_{i=1}^n$.
\end{lemma}

\begin{proof}  By elementary calculation, we have (note that ${\cal M}^2=0$, which implies $f(u,v) =f_U(u) f_0(v)$)
\begin{align}
  & \E  \big [D_i \big | {\cal F}_i \big ] = \sum_{j=1}^{i-1}\E\big [ H_n(Y_i, Y_j) \big |Y_j\big ] \\ 
    &= h^{-1} \sum_{j=1}^{i-1} \left\{\int_{\R^+}  K\left(\frac{U_i-u_j}{h}\right)f_U(u_j)  \mathrm du_j \right\} \left\{\int_{\mathbb{S}^{p-1}}c_1^{-1}(\kappa) L(\kappa V_i^{\top} v_j) f_0(v_j)  \  \omega_{p-1}(\dd v_j) - \omega_{p-1}^{-1} \right\} \\ 
    &=h^{-1}\omega_{p-1}^{-1} \sum_{j=1}^{i-1} \left\{\int_{\R^+}  K\left(\frac{U_i-u_j}{h}\right)f_U(u_j)  \mathrm du_j \right\} \left\{\int_{\mathbb{S}^{p-1}}c_1^{-1}(\kappa) L(\kappa V_i^{\top} v_j) \ \omega_{p-1}(\dd v_j) - 1 \right\}
    \label{eq:mds}
\end{align}
Using the tangent-norm representation \eqref{eq:tangent-normal}  with $y=v_j$, $x=V_i$ and \eqref{lm:measure} it follows that 
\begin{align}
     \E  \big [D_i \big | {\cal F}_i \big ]
     &=h^{-1}\omega_{p-1}^{-1} \sum_{j=1}^{i-1} \left\{\int_{\R^+}  K\left(\frac{U_i-u_j}{h}\right)f_U(u_j)  \mathrm du_j \right\} \\
     & \times  \left\{ \int_{\Omega_{V_i}}\int_{0}^{\pi} c_1^{-1}(\kappa) L(\kappa \cos \theta) (\sin \theta)^{p-2} \dd \theta \ \omega_{p-2}(\dd \xi) - 1 \right\} \\ 
     &=h^{-1}\omega_{p-1}^{-1} \sum_{j=1}^{i-1} \left\{\int_{\R^+}  K\left(\frac{U_i-u_j}{h}\right)f_U(u_j)  \mathrm du_j \right\}
     \\
     & \times
     \left\{c_1^{-1}(\kappa)  \omega_{p-2} \int_{0}^{\pi} L(\kappa \cos \theta) (\sin \theta)^{p-2} \dd \theta  - 1 \right\} = 0.
\end{align}
Since $D_i = \sum_{j=1}^{i-1}H_n(Y_i, Y_j)$ is $\mathcal F_{i}$ measurable, $
( D_i  )_{i=2}^n$ is martingale difference sequence. 
\end{proof}

\subsection{Proof of 
\cref{prop:meanvar}}
\begin{itemize}
    \item[(i)] If $\mathcal{M}^2= 0$,  the assertion follows from Lemma \ref{martingale}.
\item[(ii)] If $\mathcal{M}^2>0$, we obtain from  \eqref{eq:Ustat} and \cref{lm:kernelse} 
\begin{align}
   \E (\hat{\mathcal{M}}^2_n) &= \frac{2}{n(n-1)} \sum_{i=1}^n \sum_{j=1}^{i-1} \E H_n(Y_i, Y_j)\\  &=  \E(H_n (Y_1, Y_2))\\
   &  = \frac{1}{hc_1(\kappa)} \E \Big [  K \Big  (\frac{U_i-U_j}{h}\Big  ) L(\kappa V_i^{\top} V_j) \Big  ]  - 
    \frac{\omega_{p-1}^{-1}}{
    h}\E \Big  [  K\Big  (\frac{U_i-U_j}{h}\Big  ) \Big  ] \\
    & = \int_{\mathbb{S}^{p-1}} \int_{\mathbb R^+} f^2(u,v) \mathrm du \ \omega_{p-1}(\mathrm  dv)  - \omega_{p-1}^{-1} \int_{\mathbb R^+} f_U^2(u) \mathrm du   \\ &
   +\frac{h^2 \phi_2(K) }{2} \left\{\int_{\mathbb{S}^{p-1}} \int_{\mathbb R^+} f(u_2,v_1)\left. \frac{\partial^2 f(u, v_1)}{\partial u^2} \right|_{u=u_2} \mathrm du_2 \ \omega_{p-1}(\mathrm  dv_1)  - \omega_{p-1}^{-1} \int_{\mathbb R^+} f_U(u) f_U^{\prime \prime}(u)\mathrm du \right\}\\ 
       &+ \frac{\omega_{p-2}}{2\kappa}\int_{\mathbb{S}^{p-1}} \int_{\mathbb R^+} f(u_2, v_1)   \mathrm{tr}\{ D^2_{f}(u_2, v_1)\}      \dd u_2 \omega_{p-1}(\dd v_1)+ O(\kappa^{-3/2}+ h^3 + h^2\kappa^{-1}).
\end{align}
and the result follows by an application of \cref{lm:rate}.
\end{itemize}


\subsection{Proof of 
\cref{thm:normal}}
For the calculation of the  asymptotic variance of $\hat{\mathcal M}^2_n$,
note that for $j, k < i$ and $j \neq k$, we have 
\begin{align} 
\label{det300}
\begin{split}
\E(H_n(Y_i, Y_j) H_n(Y_i, Y_k))& =\E\{H_n(Y_i, Y_j)\E( H_n(Y_i, Y_k)|(Y_i, Y_j))\} \\
& 
=\E\{H_n(Y_i, Y_j)\E( H_n(Y_i, Y_k)|Y_i)\} = 0,
\end{split}
\end{align} 
where the last equality follows by the same arguments as given in the proof of Lemma \ref{martingale}.  
By the same Lemma  
$\hat {\cal M}_n^2$ is a sum of martingale differences, and therefore \eqref{det300} implies
\begin{align}
   s_n^2 =  \mathrm{Var}(\hat{\mathcal{M}}^2_n) &= \frac{4}{n^2(n-1)^2} \mathrm{Var}\Big ( \sum_{i=1}^n \sum_{j=1}^{i-1} H_n(Y_i, Y_j) \Big ) \\ 
    &=  \frac{4}{n^2(n-1)^2}   \sum_{i=1}^n \mathrm{Var}\Big ( \sum_{j=1}^{i-1} H_n(Y_i, Y_j) \Big ) \\ 
    &=  \frac{4}{n^2(n-1)^2}   \sum_{i=1}^n \sum_{j=1}^{i-1} \mathrm{Var} ( H_n(Y_i, Y_j)  ) \\ 
    & = \frac{2}{n(n-1)}   \mathrm{Var} ( H_n(Y_i, Y_j) ) \\ 
    & = \frac{2}{n(n-1)h^2} \Big  \{ c_1^{-2}(\kappa) \E\Big [  K^2 \Big ( \frac{U_i-U_j}{h}\Big  ) L^2 (\kappa V_i^{\top} V_j) \Big ]  \\ &  - 2 c^{-1}_1(\kappa)\omega_{p-1}^{-1} \E \Big [ K^2\Big  (\frac{U_i-U_j}{h}\Big ) L (\kappa V_i^{\top} V_j) \Big ] + 
    \omega_{p-1}^{-2}\E \Big  [ K^2\Big ( \frac{U_i-U_j}{h}\Big ) \Big ] \Big \} \\ 
    & = \frac{2}{n(n-1)h} \left[c_1^{-2}(\kappa) c_2(\kappa) \psi_2(K)  \int_{\mathbb{S}^{p-1}} \int_{\R^+} f^2(u,v) \mathrm du \ \omega_{p-1}(\mathrm  dv) \right. \\ & \left. - 2 \omega_{p-1}^{-1} \psi_2(K)  \int_{\mathbb{S}^{p-1}} \int_{\mathbb R^+} f^2(u,v) \mathrm du \ \omega_{p-1}(\mathrm  dv)  + 
    \omega_{p-1}^{-2}\psi_2(K)  \int_{\mathbb R^+} f_U^2(u) \mathrm du \right. \\ & \left. +O(c_1^{-2}(\kappa)c_2(\kappa)h^2 +  c_1^{-1}(\kappa)b_2(\kappa))\right ]\\
    & \sim  \frac{2\psi_2(K)  c_1^{-2}(\kappa)c_2(\kappa) \int_{\mathbb{S}^{p-1}} \int_{\mathbb R^+} f^2(u,v) \mathrm du \ \omega_{p-1}(\mathrm  dv) }{n(n-1)h}   \\
     & \sim  \frac{2\psi_2(K)  d^{-2}_1(p) d_2(p) \kappa^{(p-1)/2} \int_{\mathbb{S}^{p-1}} \int_{\mathbb R^+} f^2(u,v) \mathrm du \ \omega_{p-1}(\mathrm  dv) }{n(n-1)h}  , \\
\end{align}
where we have used Lemma \ref{lm:kernelse} three times in the last equality and Lemma \ref{lm:rate} for last two approximations.  

By Lemma \ref{martingale}, 
${\cal M}_n^2$ is a sum of martingale difference and 
    we can apply 
    a central limit theorem for sums of martingale differences
    \citep[Theorem 1 in ][]{HALL19841} to prove the statement in \cref{thm:normal}.
We have already shown that $H_n$ is symmetric, $\E\{H_n(Y_1, Y_2)|Y_1\} = 0$ a.s., and by  Proposition \ref{prop:meanvar} it follows that $\E\{H_n^2(Y_1, Y_2)\} < \infty$ for each $n$. Therefore it is sufficient to verify that the condition
\begin{align}
\label{det305}
    [\E\{G_n^2(Y_1, Y_2) \} + n^{-1} \E\{H_n^4(Y_1, Y_2)\} ]/[\E\{H_n^2(Y_1, Y_2)\} ]^2 \to 0 
\end{align}
is satisfied as $n \to \infty$, 
where 
\begin{align}
    G_n(x, y) = \E\big [ H_n(Y_1, x)H_n(Y_1, y) \big  ] .
\end{align}
For this purpose note that
\begin{align}
    h^4\E [ H_n^4(Y_1, Y_2) ]  &= c_1^{-4}(\kappa)\E\Big  [ K^4 \Big  (\frac{U_i-U_j}{h}\Big  ) L^4(\kappa V_i^{\top} V_j)\Big  ] - 4c_1^{-1}(\kappa)\omega_{p-1}^{-3} \E\Big  [ K^4\Big  ( \frac{U_i-U_j}{h}\Big  )  L(\kappa V_i^{\top} V_j)\Big  ] \\ 
    &+ 6c_1^{-2}(\kappa)\omega_{p-1}^{-2} \E \Big  [   K^4\Big   (\frac{U_i-U_j}{h}\Big  ) L^2(\kappa V_i^{\top} V_j) \Big  ]  \\
    & -  4c_1^{-3}(\kappa)\omega_{p-1}^{-1} \E  \Big  [ K^4\Big  (\frac{U_i-U_j}{h}\Big  ) L^2(\kappa V_i^{\top} V_j)\Big  ] + \omega_{p-1}^{-4} \E  \Big  [ K^4\Big   (\frac{U_i-U_j}{h}\Big   )\Big  ]
\end{align}
and that  
\begin{align}
    \label{det302}
\E \big  [ H_n ^2(Y_1,Y_2) \big ]  = {\rm Var}  ( H_n (Y_1,Y_2) ) \sim  
 \psi_2(K)  c_1^{-2}(\kappa)c_2(\kappa)h^{-1} \int_{\mathbb{S}^{p-1}} \int_{\mathbb R^+} f^2(u,v) \mathrm du \ \omega_{p-1}(\mathrm  dv)   ,
 \end{align}
 which follows from the calculation of  $s_n^2$.
Therefore,  Lemma  \ref{lm:rate} and  \ref{lm:kernel4se} yield
\begin{align}
\label{det306}
    n^{-1} \E \{H_n^4(Y_1, Y_2) \}/(\E(H_n^2(Y_1, Y_2))^2 = O(c_4(\kappa)/(c_2^2(\kappa) nh)) = O(\kappa^{(p-1)/2}/(nh)). \label{eq:condition1}
\end{align}
For the  calculation of  $G_n(x, y)$ we define $u_2 = \|x\|$, $v_2 = x / \| x \|$, $u_3 = \| y \|$, $v_3 = y/ \|y\|$ 
and consider 
\begin{align}
h^2 G_n(x,y)    &= h^2 \E \big  [ H_n(Y_1, x)H_n(Y_1, y) \big  ] \\
    &= \int_{\mathbb{S}^{p-1}}\int_{\mathbb R^+} \left\{c_1^{-1}(\kappa)K\left(\frac{u_1-u_2}{h}\right)L(\kappa v_1^{\T}v_2) - \omega_{p-1}^{-1} K\left(\frac{u_1-u_2}{h}\right)\right\} \\ & \times \left\{c_1^{-1}(\kappa)K\left(\frac{u_1-u_3}{h}\right)L(\kappa v_1^{\T}v_3) - \omega_{p-1}^{-1} K\left(\frac{u_1-u_3}{h}\right)\right\} f(u_1, v_1) \dd u_1 \ \omega_{p-1}(\dd v_1). \\
    & =  G_{1,n}(x,y) - G_{2,n}(x,y) - G_{3,n}(x,y) + G_{4,n}(x,y) ,
    \label{det301}
\end{align}
where $G_{1,n}(x,y) , \ldots , G_{4,n}(x,y)$ are defined by 
\begin{align}
 G_{1,n}(x,y) & =
 \int_{\mathbb{S}^{p-1}}\int_{\mathbb R^+} \left\{c_1^{-2}(\kappa)K\left(\frac{u_1-u_2}{h}\right)L(\kappa v_1^{\T}v_2) K\left(\frac{u_1-u_3}{h}\right)L(\kappa v_1^{\T}v_3)\right\}f(u_1, v_1) \dd u_1 \ \omega_{p-1}(\dd v_1)\\ 
 G_{2,n}(x,y) & =
 \omega_{p-1}^{-1} \int_{\mathbb{S}^{p-1}}\int_{\mathbb R^+}  \left\{c_1^{-1}(\kappa)K\left(\frac{u_1-u_2}{h}\right)L(\kappa v_1^{\T}v_2)K\left(\frac{u_1-u_3}{h}\right)\right\} \dd u_1 \ \omega_{p-1}(\dd v_1) \\
 G_{3,n}(x,y) & =
 \omega_{p-1}^{-1} \int_{\mathbb{S}^{p-1}}\int_{\mathbb R^+}  \left\{c_1^{-1}(\kappa)K\left(\frac{u_1-u_2}{h}\right)L(\kappa v_1^{\T}v_3)K\left(\frac{u_1-u_3}{h}\right)\right\} \dd u_1 \ \omega_{p-1}(\dd v_1)\\
 G_{4,n}(x,y) & =
 \omega_{p-1}^{-2} \int_{\mathbb R^+}K\left(\frac{u_1-u_2}{h}\right)K\left(\frac{u_1-u_3}{h}\right) f_U(u_1)  \dd u_1.
\end{align}
Note that we obtain from  \eqref{det204},
\begin{align}
    L(\kappa t) \lesssim \kappa^{(p-1)/2}
\end{align}
 for all $t \in [-1,1]$,
and with the representation 
\begin{align}
    v_1 &= v_2\cos \theta + \xi \sin \theta,
\end{align}
with   $\xi  \in \Omega_{v_2}:=\{  \xi \in \mathbb{S}^{p-1} :~ \xi \perp v_2 \}$, it follows that 
 \begin{align}
| G_{1,n}(x,y)|   
\lesssim  & \Big | h \int_{\mathbb S^{p-1}}\int_{\mathbb R^+}\Big \{c_1^{-2}(\kappa)K\left(u\right)L(v_1^{\top}v_2)K\Big  (u+\frac{u_2-u_3}{h}\Big )\kappa^{(p-1)/2} \Big\}\\ & \times  f(u h + u_2, v_1) \dd u \ \omega_{p-1}(\dd v_1)   \Big |\\
\lesssim  & \Big | h \int_{\Omega_{v_2}}\int_{0}^{\pi}\int_{-1}^1\Big \{K\left(u\right)L(\cos \theta)K\Big  (u+\frac{u_2-u_3}{h}\Big )\kappa^{(p-1)/2} \Big\}\\ & \times  f(u h + u_2, v_2\cos \theta + \xi \sin \theta)(\sin \theta)^{p-2} \dd u \dd \theta \ \omega_{p-2}(\dd \xi)   \Big |\\
\lesssim & h \kappa^{(p-1)/2}\int_{-1}^{1}\big|K(u) K \Big (u+\frac{u_2-u_3}{h} \Big ) \Big | du.
\label{eq:G1}
\end{align}
We obtain for the terms $G_{n,2}$ and $G_{n,3}$ by similar arguments the estimate
\begin{align}
    |G_{2,n}(x,y)+ G_{3,n}(x,y)|
   & \leq  h  \omega_{p-1}^{-1} c_1^{-1}(\kappa) \int_{\mathbb S^{p-1}} \int_{\mathbb R^+} \Big  | K\left(u\right)K \Big (u+\frac{u_2-u_3}{h}\Big ) \Big |\{ L(\kappa v_1^{\T}v_3)+  L(\kappa v_1^{\T}v_2)\} \dd u \ \omega_{p-1}{\dd v_1}\\ 
    & \lesssim h \kappa^{p/2-1/2} \int_{-1}^{1}\big|K(u) K \Big (u+\frac{u_2-u_3}{h} \Big ) \Big | du ,   \label{eq:G23}
\end{align} 
and finally for the term $G_{4,n}$ 
\begin{align}
     |G_{4,n}(x,y)| \leq h \omega_{p-1}^{-2} \int_{-1}^{1} \Big |K(u) K  \Big (u+\frac{u_2-u_3}{h} \Big ) \Big | du.\label{eq:G4}
\end{align}
Then, combining \eqref{eq:G1}, \eqref{eq:G23} and \eqref{eq:G4}  and using Cauchy's inequality
yields 
\begin{align}
   h^4 \E\{G_n^2(Y_1, Y_2) \} &
   \lesssim \E  [ G^2_{1,n}(Y_1, Y_2) ]  + \E \big [\{ G_{2,n}(Y_1, Y_2) + G_{3,n}(Y_1, Y_2)\}^2 \big ] +  \E [ G^2_{4,n}(Y_1, Y_2) ]\\
   &\lesssim  
  h^2 \kappa^{p-1} \int_{\mathbb R^+} \int_{\mathbb R^+} \Big \{\int_{-1}^1 \Big | K\left(u\right) K\Big (u+\frac{u_2-u_3}{h}\Big )\Big |\dd u \Big \}^2\dd u_2 \dd u_3\\
    & \lesssim   h^3 \kappa^{p-1}   \int_{-2}^2  \Big \{\int_{-1}^1  \big |K\left(u\right) K\left(u+a\right)\big |\dd u \Big \}^2 \dd a  \\ 
    & = O(h^3
    \kappa^{p-1}) . \label{eq:Gn}
\end{align}
 Finally, with  \eqref{det302}, we have  
 \begin{align}
     \E \big [ G_n^2(Y_1, Y_2) \big ]  \big /\big (\E \big [ H_n^2(Y_1, Y_2\big ] \big )^2 =   O(h^{-1} \kappa^{p-1}/ (h^{-2} \kappa^{p-1})) = O(h), 
 \end{align}
 and combining this statement with \eqref{det306} yields \eqref{det305}, which completes the proof of Theorem \ref{thm:normal}.

\subsection{Proof of \cref{thm:alter}}
We recall the definition of 
 the kernel 
\begin{align}
    H_n(Y_i, Y_j) = h^{-1} c_1^{-1}(\kappa) K\left(\frac{U_i - U_j}{h}\right) L(\kappa V_i^{\top} V_j )- h^{-1}\omega_{p-1}^{-1} K\left(\frac{U_i - U_j}{h}\right), 
\end{align}
introduce the  notation
\begin{align}
    g_n(Y_1) :=\E \big [ H_n(Y_1, Y_2)|Y_1 \big ] 
\end{align}
and   define the random variables 
\begin{align}
    g(Y_i) = f(U_i, V_i) - \omega_{p-1}^{-1} f_U(U_i), \ \ (i=1, \ldots , n)
\end{align}
At the end of the proof, we  show that 
\begin{align}
 \label{det303}   
R_n = \hat{\mathcal M}_n^2 -\E [\hat{\mathcal M}_n^2] -  \frac{2}{n}\sum_{i=1}^n \big  ( g(Y_i)- \mathcal M^2  \big )  = o_{\mathbb{P}}\Big ({1 \over \sqrt{n} } \Big ) . 
\end{align}
 Therefore, we shall first focus on deriving the limiting behavior of $\frac{1}{n}\sum_{i=1}^n (g(Y_i)- \mathcal M^2) $. By elementary calculation, we obtain
\begin{align}
     & \E [g(Y_i) ]  = \mathcal M^2~,  ~~\mathrm{Var} (g(Y_i)) = 
    \sigma^2,
\end{align}
where ${\cal M}^2$ and $\sigma^2$ are defined in \eqref{det103} and \eqref{det109}  respectively. 
Therefore, since $Y_1, \ldots , Y_n$ are independent identically distributed, the 
 Central Limit Theorem yields  
\begin{align}
  \sqrt{n}  \Big \{\frac{1}{n}\sum_{i=1}^n g(Y_i) -  \mathcal M^2\Big \} = \sqrt{n}  \Big \{\frac{1}{n}\sum_{i=1}^n \big ( g(Y_i) -  \E [ g(Y_i) ] \big ) \Big \} 
  \stackrel{d}{\longrightarrow}  \mathcal N(0, \sigma^2).\label{eq:gn}
\end{align}
Finally, we prove \eqref{det303}.    For this purpose,  note that 
\begin{align}
\label{eq:decomposition}
    R_n  &= \hat{\mathcal M}_n^2  - \frac{2}{n}\sum_{i=1}^n g(Y_i) +  2\mathcal{M}^2- \E [ H_n(Y_1, Y_2)] \\ 
    &=r_{1,n} + 2 r_{2,n},
\end{align}
where 
\begin{align}
 r_{1,n}   & = \frac{1}{n(n-1)} \sum_{i \neq j = 1}^n \big \{H_n(Y_i, Y_j)- g_n(Y_i)  - g_n(Y_j) + \E(H_n(Y_i, Y_j)) \big \}\\ 
  r_{2,n} & =  \frac{1}{n} \sum_{i=1}^n \{ (g_n(Y_i)-\E(H_n(Y_i, Y_j))) - (g(Y_i)-\mathcal{M}^2)\}.
\end{align}
By triangle inequality, we have 
\begin{align}
    \|R_n\|  := \big \{\E[ R^2_n ]  \big \}^{1/2}  \leq \|r_{1,n}\| + 2\|r_{2,n}\|. 
\end{align}
and we show at the end of the proof that 
\begin{align}
\label{det304}
\| r_{j,n}\| = o\big (   n^{-1/2}  \big ) 
~~~~ (j=1,2).\end{align} 
Then $\|R_n\|  = o(n^{-1/2})$ and, by Markov's inequality, 
$
    R_n = \op \big ( { 1 \over  \sqrt{n} } \big ) $, 
    which proves \eqref{det303}.
Combining this result with \eqref{eq:gn} yields the assertion of of Theorem \ref{thm:alter}, that is  
\begin{align}
\sqrt{n} ( \hat{\mathcal M}_n^2-\E [\hat{\mathcal M}_n^2]) \Rightarrow N(0, 4\sigma^2). 
\end{align}
The proof is now  completed proving the estimates in \eqref{det304}.
To derive a corresponding estimate for $r_{1,n}$ we note that, by conditioning on $Y_i,Y_j$,  
$$
\mathrm{Cov}\big ( \big  \{H_n(Y_i, Y_j)- g_n(Y_i) - g_n(Y_j) + \E(H_n [(Y_i, Y _j) ]  \big  \} , \big  \{H_n(Y_i, Y_k)- g_n(Y_i) - g_n(Y_k) + \E [ H_n(Y_i, Y_j)] \big  \}] =0 ,
$$
whenever the indices $i,j,k$ are different from each other. Consequently, we obtain from 
 \cref{prop:meanvar}
\begin{align}
\mathrm{Var}(r_{1,n}) &= \frac{1}{n(n-1)}   \mathrm{Var}\{H_n(Y_i, Y_j)- g_n(Y_i) - g_n(Y_j) + \E [ H_n(Y_i, Y_j)] \}  \\ 
    &+  \frac{n-2}{n(n-1)} \mathrm{Cov}[\{H_n(Y_i, Y_j)- g_n(Y_i) - g_n(Y_j)  + \E [ H_n(Y_i, Y_j)] \} , \\
    &\{H_n(Y_i, Y_k)- g_n(Y_i) - g_n(Y_k) + \E [ H_n(Y_i, Y_j) ] \}]\\
    &= \frac{1}{n(n-1)} \big [\mathrm{Var}\{H_n(Y_i, Y_j)\} - 2\mathrm{Var} \{g_n(Y_i)\} \big ] \\ 
    & = O \Big (s_n^2 + \frac{1}{n(n-1)} \sigma^2 \Big ) = O(\kappa^{(p-1)/2} n^{-2} h^{-1} + n^{-2}) = o(n^{-1}), \label{eq:varRn}
\end{align}
where the second equality follows by a tedious calculation observing that  that 
\begin{align}
    \E \{H_n(Y_i, Y_j) H_n(Y_i, Y_k)\} &= \E \{ H_n(Y_i, Y_j) \E(H_n(Y_i, Y_k)|(Y_i, Y_j))\} \\ &= \E \{H_n(Y_i, Y_j) \E(H_n(Y_i, Y_k)|Y_i)\}\\ &= \E \{H_n(Y_i, Y_j) g_n(Y_i)\}\\ &= \E (\E( H_n(Y_i, Y_j)|Y_i) \E\{H_n(Y_i, Y_k)|Y_i)\} = \E(g_n^2 (Y_i)) ,
\end{align}
whenever the indices $i,j,k$ are different.
As $r_{1,n}$ is centered, \eqref{det304} follows for $j=1$. 
Finally, since $Y_i$'s are $i.i.d.$, we have 
\begin{align}
    \|r_{2,n}\|\lesssim n^{-1/2} \| (g_n(Y_i)-\E(H_n(Y_i, Y_j))) - (g(Y_i)-\mathcal{M}^2)\|  = O(n^{-1/2}(1/\kappa + h^2)) = o(n^{-1/2}).
\end{align}

\subsection{Proof of \cref{thm:invariance}}
Recall the definition of the process $S_n$ in \eqref{det207a}, define the random variables 
\begin{align}
    g(Y_i) = f(U_i, V_i) - \omega_{p-1}^{-1} f_U(U_i)
\end{align}
($i=1, \ldots , n$)
and consider the stochastic process 
\begin{align}
    S_n^{\circ}(t) = \frac{1}{\sqrt{n}} \sum_{i=1}^{\lfloor nt \rfloor} (g(Y_i)-  \mathcal M^2)  .
\end{align}
The assertion of Theorem \ref{thm:invariance} is now proved in two steps. 
\medskip

\noindent
\textbf{Step 1}: we show that 
\begin{align}
    \max_{
    {2} \leq k \leq n} {1 \over \sqrt{n} } \Big |k (\hat{\mathcal M}_k^2  -  \E [\hat{\mathcal M}_n^2] )- 2\sum_{i=1}^k (g(Y_i)  -  \mathcal M^2 ) \Big |  = o_{\mathbb{P}}(1), \label{eq:approx}
\end{align}
which implies 
\begin{align}
    \sup_{t \in [0,1]}|S_n(t) - S_n^{\circ}(t)| = o_{\mathbb{P}}(1).
\end{align}
\textbf{Step 2}: we  prove that 
\begin{align}
  \big \{  S_n^{\circ}(t) \big \}_{t \in [0,1]}  \Rightarrow \big \{ \sigma \mathbb{B} (t) \big \}_{t \in [0,1]}
\end{align}
in $\ell^\infty ([0,1])$,
where $\big \{ \mathbb{B} (t) \big \}_{t \in [0,1]} $ is a standard Brownian motion and $\sigma^2 $ is defined in \eqref{det109}.
 \medskip

\noindent
\textbf{Proof of Step 1}:
Recall that 
\begin{align}
    g_n(Y_1) = \E [ H_n(Y_1, Y_2)| Y_1) ] , 
\end{align}
and note  that 
\begin{align}
    R_{k,n}^{\circ}  &= k(\hat{\mathcal M}_k^2 -  \E [\hat{\mathcal M}_n^2] )  - 2\sum_{i=1}^k (g(Y_i)  -  \mathcal{M}^2 )
    \label{eq:R1R2} \\ 
    & = \frac{1}{(k-1)} \sum_{i \neq j = 1}^k \Big \{H_n(Y_i, Y_j)- g_n(Y_i) - g_n(Y_j)   +  \E [ H_n(Y_i, Y_j) ]\\& +  g_n(Y_i)-\E [ H_n(Y_i, Y_j)]   + g_n(Y_j)-\E [ H_n(Y_i, Y_j)]  - (g(Y_i) -\mathcal M^2) - (g(Y_j)-\mathcal M^2)\Big \}.
    \nonumber
\end{align}
For a real valued random variable $Z$ we define  $\|Z \|$ denote $\{\E [ Z^2 ] \}^{1/2}$, then we obtain 
\begin{align}
    \Big \|\max_{2 \leq k \leq n}|R_{k,n}^{\circ}|  \Big
    \| & \leq   
    R_1 +2  R_2,
\end{align}
where $R_1$ and $R_2$ are defined  by

\begin{align}
     R_1 & =  \Big \|\max_{2 \leq k \leq n} \Big |\frac{1}{(k-1)} \sum_{i \neq j = 1}^k \Big \{H_n(Y_i, Y_j)- g_n(Y_i) - g_n(Y_j) + \E  [ H_n(Y_i, Y_j) ] \Big  \} \Big | \Big \|\\ 
  R_2   & =\Big \|\max_{2 \leq k \leq n} \Big |  \sum_{i= 1}^k \big  \{g_n(Y_i) - \E [ H_n(Y_i, Y_j)] - (g(Y_i) - \mathcal M^2) \big \}\Big | \Big \|.
\end{align}
Similar arguments as given in the proof of Lemma \ref{martingale} show that the random variables 
$$
\sum_{j=1}^{i-1} \big \{H_n(Y_i, Y_j)- g_n(Y_i) - g_n(Y_j)   + \E [ H_n(Y_i, Y_j)  ] \big  \}
$$
are martingale differences  with respect to the filtration $ ( \mathcal F_i)_{i=1,\ldots , n} $, where
$\mathcal F_{i} = \sigma (Y_1 \ldots, Y_{i})$ is the sigma field generated by $Y_1 \ldots, Y_{i}$. Moreover,  using  similar arguments as in  the calculation in \eqref{eq:varRn}, we obtain 
\begin{align}
& \mathrm{Var}\Big ( \sum_{j=1}^{i-1} \big  \{H_n(Y_i, Y_j)- g_n(Y_i) - g_n(Y_j)    + \E [H_n(Y_i, Y_j) ]  \big \}\Big )  \\  & = \sum_{j=1}^{i-1}  \mathrm{Var}\big ( \{H_n(Y_i, Y_j)- g_n(Y_i) - g_n(Y_j)    + \E [H_n(Y_i, Y_j) ]  \big )  \\ 
& + \sum_{j\neq j^{\prime}}^{i-1} \mathrm{Cov}  \Big  ( \big \{ H_n(Y_i, Y_j)- g_n(Y_i) - g_n(Y_j)    + \E [ H_n(Y_i, Y_j)]  \big \},  \\
& ~~~~~~~~~~~~~~~~~~~~~~~~~~~~~~~~~  \big \{H_n(Y_i, Y_{ j^{\prime}})- g_n(Y_i) - g_n(Y_{ j^{\prime}})   + \E [ H_n(Y_i, Y_{ j^{\prime}})]  \big \}\Big ) \\ 
& =  (i-1) [\mathrm{Var}\{H_n(Y_1, Y_2)\} - 2\mathrm{Var} \{g_n(Y_1)\}] \\
& = O((i-1) (\kappa^{(p-1)/2} h^{-1} + \sigma^2)).\label{eq:varR1}
\end{align}
For a constant $\rho > 0$, $c = \lfloor 1/(\log \rho)\rfloor + 1$, we have 
\begin{align}
    R_1 &\leq \Big \|\max_{1 \leq l \leq \lfloor c\log n \rfloor} \max_{\rho^{l-1} \leq k \leq \rho^{l}} \Big|\frac{1}{(k-1)} \sum_{i \neq j = 1}^k \big \{H_n(Y_i, Y_j)- g_n(Y_i) - g_n(Y_j)   + \E [H_n(Y_i, Y_j) ] \big \}\Big | \Big \|\\
    & \leq \sum_{l=1}^{\lfloor c\log n \rfloor}(\rho^{l-1}-1)^{-1}  \Big \|\max_{1 \leq k \leq \rho^{l}} \Big | \sum_{i \neq j = 1}^k \{H_n(Y_i, Y_j)- g_n(Y_i) - g_n(Y_j)   + \E [ H_n(Y_i, Y_j)]  \}\Big | \Big \|\\
    & = O \Big ( \sum_{l=1}^{\lfloor c\log n \rfloor}\rho^{-l+1} \Big \|\sum_{i=1}^{\rho^l}\sum_{j=1}^{i-1} \{H_n(Y_i, Y_j)- g_n(Y_i) - g_n(Y_j)   + \E [ H_n(Y_i, Y_j)] \} \Big \| \Big)\\
     & = O \Big ( \sum_{l=1}^{\lfloor c\log n \rfloor}\rho^{-l+1} \Big \{\sum_{i=1}^{\rho^l}\Big  \|\sum_{j=1}^{i-1} \{H_n(Y_i, Y_j)- g_n(Y_i) - g_n(Y_j)   + \E [H_n(Y_i, Y_j)]  \}\|^2\Big \}^{1/2} \Big )\\ 
    & = O\Big (\sum_{l=1}^{\lfloor c\log n \rfloor}\rho^{-l+1}\Big \{\sum_{i=1}^{\rho^l} (i-1) (\kappa^{(p-1)/2} h^{-1} + \sigma^2) \Big \}^{1/2} \Big )\\
    &  = O \big (\log n (\kappa^{(p-1)/2} h^{-1})^{1/2}) \big )   = o(\sqrt{n}), \label{eq:R1rate}
\end{align}
where the first equality follows from Doob's inequality, the  second follows from the fact that $\sum_{j=1}^{i-1} \big \{H_n(Y_i, Y_j)- g_n(Y_i) - g_n(Y_j)  + \E [ H_n(Y_i, Y_j)]  \big \}$ are martingale differences and the third follows from \eqref{eq:varR1}. 
 Similarly, applying Doob's inequality for the $i.i.d.$ sequence $g_n(Y_i) - g(Y_i)$ 
\begin{align}
    R_2 & \leq   \Big \|  \sum_{i = 1}^n\{g_n(Y_i) - \E [ H_n(Y_i, Y_j)] - (g(Y_i) - \mathcal M^2)  \}  \Big \|  \\
    & \lesssim \sqrt{n} \|g_n(Y_i) - \E [ H_n(Y_i, Y_j)] - (g(Y_i) - \mathcal M^2)\| \big ) \\ &
 \lesssim
\sqrt{n} R_{21} + \sqrt{n} R_{22}+ \sqrt{n} R_{23}, \label{eq:R2decompos}
\end{align}
where 
\begin{align}
   R_{21}   & =\Big \| \int_{\mathbb R^+} \int_{\mathcal \mathbb{S}^{p-1}}h^{-1} c_1^{-1}(\kappa)  K\Big (\frac{u_2 - U_1}{h}\Big ) L(\kappa v_2^{\top} V_1) f (u_2, v_2) \dd u_2 \  \omega_{p-1}(\dd v_2) - f (U_1, V_1) \Big \| , \\ 
   R_{22}  & =\Big \|\omega_{p-1}^{-1} \Big \{ \int_{\mathbb R^+}h^{-1} K \Big (\frac{u_2 - U_1}{h}\Big ) f_U(u_2) \dd u_2 - f_U(U_1)\Big \}\Big \| ~, \\
   R_{23} & =| \E [ H_n(Y_i, Y_j)] - \mathcal M^2|. 
    \end{align}
A    Taylor expansion gives  for the first term 
\begin{align}
    R_{21} & = \Big \| \int_{-1}^1 \int_{\mathcal \mathbb{S}^{p-1}}  c_1^{-1}(\kappa) K\left(u\right) L(\kappa v_2^{\top} V_1) f (U_1+hu, v_2) \dd u \  \omega_{p-1}(\dd v_2) - f (U_1, V_1) \Big \|\\
     &\leq\Big \|  \int_{-1}^1 \int_0^{\pi} \int_{\Omega_{v_1}}   c_1^{-1}(\kappa) K\left(u\right) L(\kappa 
     \theta )  \\
     & ~~~~~~~~~~~~~~~~~~  \times f (U_1+hu, V_1 \cos \theta + \xi \sin \theta) (\sin \theta)^{p-2} \dd u \ \dd \theta \ \omega_{p-2}(\dd \xi) - f (U_1, V_1) \Big \|\\ 
     &\lesssim  \Big \|  \int_{-1}^1 \int_0^{\pi} \int_{\Omega_{V_1}}  c_1^{-1}(\kappa) K\left(u\right) L(\kappa 
     \theta ) \big \{ f (U_1+hu, V_1 \cos \theta + \xi \sin \theta)  \\ & ~~~~~~~~~~~~~~~~~~
     - f (U_1, V_1 \cos \theta + \xi \sin \theta) \big \}  (\sin \theta)^{p-2} \dd u \ \dd \theta \ \omega_{p-2}(\dd \xi) \Big \|\\ 
    & + \Big \|  \int_{-1}^1 \int_0^{\pi} \int_{\Omega_{V_1}} c_1^{-1}(\kappa) K\left(u\right) L(\kappa 
     \theta ) \big   \{ f (U_1, V_1 \cos \theta + \xi \sin \theta) \\
     &  ~~~~~~~~~~~~~~~~~~
    - f (U_1, V_1) \big \} (\sin \theta)^{p-2} \dd u \ \dd \theta \ \omega_{p-2}(\dd \xi) \Big \|\\ 
     & \lesssim h +   \int_{-1}^1 \int_0^{\pi} \int_{\Omega_{V_1}}  c_1^{-1}(\kappa) K\left(u\right) L(\kappa 
     \theta )  \theta  (\sin \theta)^{p-2} \dd u \ \dd \theta \ \omega_{p-2}(\dd \xi) \\
     & = O \big (h + b_1(\kappa)/c_1(\kappa) \big ) = O(h + \kappa^{-1/2}), 
\end{align}
where we used the fact that  $\partial f(u,v)/ \partial u$  and  $\mathcal D_{f}(u,v)$ are  uniformly bounded and Lemma \ref{lm:rate}. By  similar but simpler arguments  we obtain  $R_{22} = O(h)$. By \cref{prop:meanvar}, we have, $R_{23} = O(1/\kappa+ h^2)$.  
Combining  these estimates with \eqref{eq:R2decompos} yields 
\begin{align}
    R_2 = O(\sqrt{n}(h + \kappa^{-1/2})) = o(\sqrt{n}). \label{eq:R2rate}
\end{align} 
Finally, \eqref{eq:approx} follows from \eqref{eq:R1R2}, \eqref{eq:R1rate}, \eqref{eq:R2rate}.
\medskip

\noindent
\textbf{Proof of Step 2}: 
Note that  the random variables $g(Y_i)$ are $i.i.d.$  with 
\begin{align}
    \E [ g(Y_1) ] = \mathcal M^2 ~~~ \text{ and  } ~~\mathrm{Var}(g(Y_1)) = \sigma^2.
\end{align}
Therefore, by Donsker's Theorem 
\citep[see the discussion on page 225 - 226 in][]{van1996weak} it follows  that
\begin{align}
  \big \{  S_n^{\circ}(t) \big \}_{t \in [0,1]}  \Rightarrow \big \{ \sigma \mathbb{B} (t) \big \}_{t \in [0,1]}
\end{align}
in $\ell^\infty ([0,1])$.

\bibliographystyle{apalike}
\bibliography{main}

\end{document}